\documentclass[preprint,12pt,authoryear]{elsarticle}

\usepackage{amsmath, amssymb, enumerate}
\usepackage{amsthm}
\usepackage{graphicx}
\usepackage{mathrsfs}
\usepackage{txfonts}
\usepackage{float}
\usepackage{nicematrix}
\usepackage{booktabs}
\usepackage{tikz}
\usepackage{multirow}
\usetikzlibrary{arrows.meta, positioning, calc}
\usepackage{mathtools}
\usepackage[shortlabels]{enumitem}
\usepackage{stmaryrd}
\usepackage{xcolor}
\definecolor{linkcolor}{RGB}{87,41,40}

\usepackage{hyperref}

\hypersetup{
    colorlinks=true
}

\AtBeginDocument{
  \hypersetup{
    linkcolor=linkcolor,
    citecolor=linkcolor,
    urlcolor=linkcolor
  }
}

\usepackage[nameinlink,capitalise]{cleveref}

\newcommand{\norminf}{\mathop{\inf\vphantom{\sup}}}
\newcommand{\normsup}{\mathop{\sup\vphantom{\inf}}}

\newcommand{\proj}{\operatornamewithlimits{proj}}

\newcommand{\sta}{\operatornamewithlimits{st}}
\newcommand{\stan}{\operatornamewithlimits{st}}
\newcommand{\sgn}{\operatornamewithlimits{sgn}}

\newcommand{\nstar}{{}^*}
\newcommand{\nstarback}{{{}^*\!}}
\newcommand{\nstarless}{ {{}^*\!\!} }
\newcommand{\reals}{{\mathbb{R}}}

\newcommand{\integers}{{\mathbb{N}}}
\newcommand{\ppi}{\times_{i \in I}}
\newcommand{\eqfin}{Eq^{Fin}}
\newcommand{\eqiwu}{Eq^{iwu}}

\newcommand{\mcalF}{\mathcal{F}}

\newcommand{\mcalJ}{\mathcal{J}}

\newcommand{\mcalP}{\mathcal{P}}

\newcommand{\mcalX}{\mathcal{X}}

\newcommand{\closure}{\operatornamewithlimits{cl}}
\newcommand{\conv}{{\mbox{co}}}

\newcommand{\ba}{\boldsymbol{a}}
\newcommand{\bv}{\boldsymbol{v}}
\newcommand{\bx}{{\boldsymbol{x}}}

\newcommand{\by}{{\boldsymbol{y}}}

\newcommand{\ii}{{i\in I}}
\newcommand{\nn}{{n \in \integers}}

\newcommand{\rar}{\rightarrow}

\newcommand{\half}{\frac{1}{2}}

\let\IG\iffalse

\makeatletter
\newtheoremstyle{notindented}
  {8pt}
  {8pt}
  {\addtolength{\@totalleftmargin}{2.5em}
   \addtolength{\linewidth}{-3.5em}
   \parshape 1 1.5em \linewidth}
  {}
  {\bfseries}
  {}
  {.5em}
  {\thmname{#1}\thmnumber{ #2.}  \thmnote{--- \scshape #3.}
}
\makeatother

\theoremstyle{plain}
\newtheorem{thmm}{Theorem}

\newtheorem{corro}{Corollary}[thmm]
\newtheorem{lemma}{Lemma}[section]

\makeatletter
\newtheoremstyle{defnotindented}
  {8pt}
  {8pt}
  {\addtolength{\@totalleftmargin}{3.5em}
   \addtolength{\linewidth}{-3.5em}
   \parshape 1 1em \linewidth}
  {}
  {\bfseries}
  {.}
  {.5em}
  {\thmname{#1}\thmnumber{ #2}\thmnote{--- \scshape #3.}}
\makeatother

\theoremstyle{defnotindented}
\newtheorem{dfnn}{Definition}[section]
\newtheorem{exx}{Example}[section]

\crefname{thmm}{Theorem}{Theorems}
\crefname{corro}{Corollary}{Corollaries}
\crefname{lemma}{Lemma}{Lemmas}
\crefname{dfnn}{Definition}{Definitions}
\crefname{exx}{Example}{Examples}

\journal{Journal of Economic Theory}

\begin{document}

\begin{frontmatter}

\title{All Games Have Equilibria\tnoteref{t1}}

\tnotetext[t1]{The authors gratefully acknowledge with warm gratitude Sam Alexander, Justin Bledin, Michael Grossberg, Joe Lawlor, Rich McLean, Larry Moss, Rohit Parikh, Kris Patel, Andrew Powell, Kevin Reffett, Phil Reny, Eddie Schlee, Joel Sobel, Jack Stecher, and Metin Uyanik. Khan should also acknowledge Jeremy Goodman’s workshop \textit{Topics in Epistemology} in the spring of 2024. Preliminary versions of this paper were presented by Stinchcombe at seminars at UT Austin and the University of California at Riverside; he thanks participants for their active engagement. This version is to be presented at the \textit{Fourth International Workshop on Game Theory and Economic Applications} to be held at the University of S\~{a}o Paulo, from July 26 to August 2, 2026.}

\author[add1]{M.~Ali Khan\corref{cor1}}
\ead{akhan@jhu.edu}

\author[add2]{Arthur Paul Pedersen}
\ead{appedersen@gc.cuny.edu}

\author[add3]{Maxwell B.~Stinchcombe}
\ead{max.stinchcombe@gmail.com}
\cortext[cor1]{Corresponding author}

\affiliation[add1]{organization={Department of Economics, Johns Hopkins University}}
\affiliation[add2]{organization={Department of Computer Science, The City College of New York \& The Graduate Center,\\ The City University of New York}}
\affiliation[add3]{organization={Department of Economics, the University of Texas at Austin}}

\begin{abstract}
Research on Nash equilibrium existence for infinite games has grown into a patchwork of technical preconditions and counterexamples. This paper presents a unified program in equilibrium theory by revising the predominant model of mixed strategies based on countable additivity. A game is specified by a nonempty set of players and, for each player, a nonempty action set and a bounded von Neumann-Morgenstern utility function. Every such game is shown to admit a Nash equilibrium in finitely additive mixed strategies. In addition, the equilibrium correspondence for any such game is shown to be nonempty, compact-valued, and upper hemicontinuous, and the same is true for equilibria obtained as limits of finite approximations. Techniques developed in this paper show that infinite games long treated as intractable become amenable to direct equilibrium analysis.
\end{abstract}

\begin{keyword}
Infinite games \sep
Finitely additive mixtures \sep
Equilibrium existence \sep
Finite approximations \sep
Iterative deletion of weakly dominated strategies
\end{keyword}

\end{frontmatter}
\tableofcontents

\clearpage


\section{Introduction}

A game is finite if there is a finite set of players and each player has a
finite set of actions.
For finite games, Nash's existence theorem tells us
that there exists a vector of strategies, possibly mixed, one for each player,
that is immune to pure strategy deviations.
This is a foundational result for
game theory, but, until now, there has been no result of comparable generality
available for infinite games.

A game is infinite if the player set, or their action sets, or both, are
infinite.
The mathematical framework that has been used to study equilibrium existence for
infinite games is one of overpowering elegance, power and reach.
It
encompasses material on the interplay between topology and measure, the
integrals of measurable functions, including measurable selections, taking
values in infinite dimensional vector spaces, and deep fixed point theorems.
Despite its manifold virtues, it is the wrong framework for a general theory
of games.

The conventional path to equilibrium existence in infinite games runs through
compactness and joint continuity.
\citet{ville1938jeux} showed that zero-sum
games with jointly continuous utility functions played on products of the unit
interval have countably additive equilibria.
\citet{glicksberg1952further}
extended this to finite-player games with jointly continuous payoffs and
compact Hausdorff action spaces.
Compactness and joint continuity are
essential to these results:  \citet[\S 1]{wald1945generalization} gave a
non-compact zero-sum game with no approximate countably additive equilibria;
and \citet{sion1957game} gave a zero-sum game on the unit square with
discontinuities that also preclude the existence of approximate countably
additive equilibria.
In the intervening decades, considerable effort has been
devoted to identifying conditions weaker than joint continuity and compactness
that are still sufficient for the existence of countably additive equilibria.
We replace the
patchwork of results applicable to subclasses of infinite games with a single
overarching program in equilibrium theory.

For us, a game is specified by a nonempty set of players of arbitrary
cardinality and, for each player, a nonempty set of actions and a bounded von
Neumann–Morgenstern utility function defined on the product of the spaces of
actions.
This replaces action spaces endowed with special topological or
measure theoretic structures with arbitrary spaces.
This replaces von
Neumann-Morgenstern utility functions assumed to have continuity properties
everywhere up to some limited set of exceptional points.
This replaces the
model of mixed strategies having special domains, either Borel or Baire
$\sigma$-fields, and special limit properties with probabilities defined on
all subsets and satisfying $p(A \cup B) = p(A) + p(B)$ for disjoint sets.
For
games defined at this level of generality, there always exists a vector of
mixed strategies that is immune to pure strategy deviations.
\medskip

\noindent \textbf{\cref{thm:equilibriaexist}}.  All games have finitely additive
mixed-strategy equilibria.

\medskip

Despite the well-known simplicity of their definitions and the perhaps less
well-known simplicity of the integration theory that goes with them, finitely
additive mixtures are often understood as having ``various peculiar
properties'' (e.g.\ \citet[p. 56]{yosida1952finitely}).
One might worry that
the use of such probabilities would deliver an equilibrium theory that is
either ill-behaved or unusable.
We define the convergence of probabilities by
asking that integrals against any and all bounded functions converge.
With
this, we show that the set of finitely additive equilibria is well-behaved.
\medskip

\noindent \textbf{\cref{thm:closedeqmset}}.
The equilibrium correspondence from utility
functions to the set of finitely additive equilibria is nonempty valued,
compact valued, and upper hemicontinuous.
\medskip

\cref{thm:closedeqmset} thereby delivers \textit{stability} in the form of a nonempty, compact-valued, and upper hemicontinuous equilibrium correspondence. Existence and stability, however, do not constitute a complete analysis of infinite games. An equilibrium concept is expected to be  \textit{operationalizable}: it must connect the equilibria of infinite games to equilibria that can be obtained, approximated, and analyzed in 
large but finite models. Equilibrium concepts for infinite games should be justified not only by their existence and stability properties, but also by their relationship to equilibria in finite approximations to the games and to the strategic logic of those games. The central difficulty with the classical countably additive framework is thus not merely that equilibria may fail to exist even when approximate equilibria do exist, but that passing to the countably additive limit can erase payoff-relevant distinctions that are present in finite approximations.
While these distinctions may disappear in the limit, they can be decisive for equilibrium behavior in the finite models that the infinite game is meant to represent.

These considerations motivate the third contribution of the paper.  Rather
than treating finite approximability as a subsidiary refinement, we take it to
be a structural requirement on equilibrium analysis itself.
We therefore
isolate two subcorrespondences of the equilibrium correspondence.  The first
consists of \textit{finitely approximable equilibria}: limits of equilibria or
approximate equilibria on finite approximations to the action spaces.
The
second refines this further by requiring that the approximating equilibria
survive iterated deletion of weakly dominated strategies.
These latter
equilibria are not only stable, but \textit{procedurally interpretable}: they
arise from explicit finite models, respect dominance reasoning at every stage,
and retain precisely the payoff-relevant information present in those models.
\medskip

\noindent \textbf{\cref{thm:iterundtd}}.  The equilibrium correspondences from utility
functions to (1) the set of finitely approximable equilibria or to (2) the
subset of finitely approximable equilibria that are limits of iteratively
undominated equilibria are nonempty valued, compact valued, and upper
hemicontinuous.
Further, the limits in (2) put no mass on the set of 
iteratively weakly dominated strategies.
\medskip

These three theorems define a single program. \cref{thm:equilibriaexist} 
delivers \textit{equilibrium existence}: every game with bounded payoffs 
admits an equilibrium in finitely additive mixed strategies, independent 
of any topological or measurable structure on the action spaces. 
\cref{thm:closedeqmset} delivers \textit{equilibrium stability}: the 
equilibrium correspondence is nonempty, compact-valued, and upper 
hemicontinuous under perturbations of payoffs. \cref{thm:iterundtd} delivers \textit{equilibrium operationalizability}: 
it selects equilibria that arise as limits of equilibria in finite 
approximations and that respect dominance reasoning along those 
approximations. Infinite games long treated as intractable become amenable 
to direct equilibrium analysis.

The unconventional path that achieves all this is, we contend, the natural
one.
Relaxing countable additivity to finite additivity is often dismissed as
a merely technical maneuver.  It is not.
The distinction matters for game
theory because it determines what information mixed strategies are capable of
retaining in games with infinite action sets.
In discontinuous problems, the
countably additive framework can systematically discard information encoded in
finite approximations, and this information is relevant for payoff comparisons
and strategic responses.
Finite additivity preserves this information by
evaluating limits in a way that remains faithful to the payoff functions.
Yet
the case for or against countable additivity cannot be settled in the abstract
--- it must be informed by the uses to which the probabilities are put.
For
games with infinite action sets, finite additivity is the better choice.

Our model of mixed-strategies is not indiscriminately revisionist.
Where the
classical theory already works, finite additivity does not overturn it.
For
compact Hausdorff action spaces games with arbitrary player sets and jointly
continuous payoffs, finitely additive equilibrium utilities coincide with
countably additive equilibrium utilities (\cref{cor:generalizedGlicksberg}), generalizing \citet{glicksberg1952further}.
Many games admit
countably additive $\epsilon$-equilibria for all $\epsilon > 0$ but no exact
countably additive equilibrium.
Where the classical theory fails, finite
additivity fills the gap.  Finitely additive limits of these approximate
equilibria exist and are themselves equilibria (\cref{cor:caepsilonequilibria}).
Moreover, the standard technical objection associated with
finitely additive strategies --- the failure of Fubini's theorem --- is not relevant for the equilibria singled out by \cref{thm:iterundtd}: limits of products of finitely supported probabilities remain product extensions (\cref{lem:fubiniismoot}). The insistence on product extensions separates the present theory from finitely additive correlated equilibrium constructions such as those in \citet{hart1989existence}, where the equilibrium object is a joint distribution and independence across players is not imposed.

The scope of the revision is made more precise in \cref{sec:specialdiscontinuities}.
\cref{thm:ctsequivofacaeqm,thm:ctsequivofafaeqm} identify conditions under which finitely additive
equilibria and countably additive equilibria are utility-equivalent, and
conditions under which they are not.
When payoff discontinuities are
negligible in an appropriate sense, or when deviations are evaluated
continuously at the limit, the finitely additive equilibrium set collapses
back to the classical utility predictions.
When those conditions fail, the
divergence reflects payoff-relevant distinctions present in finite
approximations that the countably additive framework cannot represent.

\cref{sec:antecedentlit} surveys antecedent work on equilibrium
existence in infinite games, organized around four approaches: finitely
additive strategies; relaxation of compactness;
restrictions to well-behaved
discontinuities; and sharing-rule equilibria.  It identifies the structural
limitations of each and explains why none delivers a general existence theory
that is faithful to the strategic structures in the games.
Our approach
delivers this.

\cref{sec:eqandproperties} presents the framework underlying \cref{thm:equilibriaexist,thm:closedeqmset,thm:iterundtd}. It defines the class of games, the space of finitely additive
probabilities, their compactness and convergence properties, and the
equilibrium concept based on immunity to pure deviations.
It then establishes
existence, continuity of the equilibrium correspondence, and the finitely
approximable and dominance-based refinements.

\cref{sec:gamesontheline} applies the framework to several
classical two-person games on the line that have served as benchmarks for
nonexistence and discontinuity phenomena.
Each example is analyzed via finite
approximations and iterated deletion of weakly dominated strategies, making
explicit how finitely additive limits retain payoff-relevant information
present in the finite models.

\cref{sec:specialdiscontinuities} investigates when finitely additive
and countably additive equilibria yield the same utility predictions, thereby
delineating the boundary between recovery of the classical theory and genuine
enlargement of it.

\cref{sec:summary} draws methodological conclusions and discusses
extensions to continuum extensive-form games and infinite-player models, where
loss of information in classical limits is not incidental but central to the
equilibrium problem.

\section{Antecedent Literature}
\label{sec:antecedentlit}

Absent compactness and joint continuity, standard equilibrium existence
arguments break down.
The literature has responded to this challenge by
adopting one of four distinct approaches.
The first replaces compact action
sets with measure spaces of actions and substitutes countably additive
strategies with the larger class of finitely additive mixed strategies.
The
second retains countable additivity but relaxes compactness.  The third and
fourth approaches retain both compactness and countable additivity, but differ
in the scope and extent to which they treat utilities exhibiting
discontinuities.
The third approach restricts attention to utilities whose
discontinuities still allow for equilibrium existence in countably additive
strategies.
The fourth goes further, allowing for arbitrary discontinuities
but using a limit and convexification process to redefine the utility function
at discontinuity points.
We discuss how our approach improves on all of these by filling in the 
links in Figure~1, assessing each against the three requirements 
identified in the introduction: existence, stability, and 
operationalizability.

\begin{figure}[ht]
\centering
\includegraphics[scale=.65]{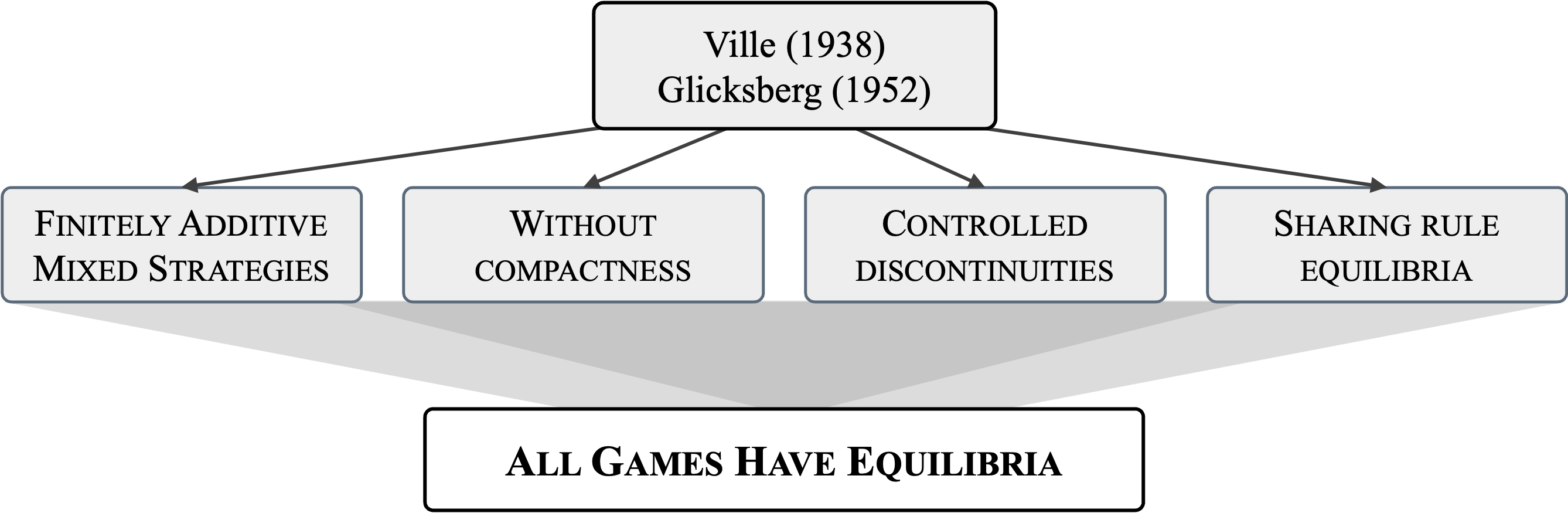}
\caption{Overview of the literature}
\end{figure}

\subsection{Finitely Additive Mixed Strategies}\label{subsect fa strats}

Much of the prior work with finitely additive strategies is of doubtful merit
for reasons both obvious and hidden.
In the obvious cases, central claims are
categorically false.  In the other cases, some of the work has implicitly confined
attention to games isometrically isomorphic to compact games with jointly
continuous payoffs, nullifying any presumption to greater generality achieved
by way of finitely additive strategies.
Other parts of this work have adopted theories of
integration that systematically skew payoffs against deviations to deliver 
``equilibria'' with payoffs that are not close to being feasible.
Still other parts of this work
redefine ``equilibrium'' in ways that risk compromising a game's
strategic integrity or fail to determine equilibrium utility levels.
The
remaining work makes use of correlating devices to determine payoffs that may
skew payoffs against deviations and thereby deliver equilibria unrelated to
the games' strategic structures.

\citet[Theorem 10, p.\ 152-3]{karlin1950operator} claims that for any bounded
measurable kernel $k(x, y)$ on $X \times Y = [0,1] \times [0,1]$, there exist
finitely additive probabilities $\mu_1$ and $\mu_2$ constituting a Nash
equilibrium for the zero-sum game with strategy sets $X = Y = [0,1]$ and
utilities $(k(x,y), -k(x,y))$.
The claim is false.  This is demonstrated in
\S\ref{subsection Karlin's claims} which shows, \textit{inter alia}, how far
off one goes by trying to use iterated integrals to define payoffs.
\citet[p.\ 153]{yanovskaya1970solution} astutely traces the source of Karlin's
error to the implicit use of Fubini's theorem.
As this fails for finitely
additive measures, the order of integration in Karlin's optimization problems
matters for the payoffs.
She argues that, since the order of integration in
the players' optimization problems may determine payoffs, one should interpret
the result as yielding existence for the sequential games with one or the
other player moving first and the other responding.

In pursuing equilibrium existence within the finitely additive framework, the
literature has had to contend with the failure of Fubini's theorem.
The
conditions imposed to avoid that failure have, necessarily, returned the
analysis to the \citet{glicksberg1952further} setting of compact
metric spaces and jointly continuous payoffs.
Some background for this
statement is in order.  

\citet[Theorem 2.2, p.\ 567]{ptak1964extension} establishes criteria for the
joint continuity of a separately continuous function on a product space.
\citet[Theorem 7]{simons1968iterated} and \citet[Theorem
4.4]{sinclair1974finitely} deploy this result to characterize functions for
which there is order-independent equality of iterated integrals for all
finitely additive probabilities.
\citet[Theorem 5, p.\ 199]{young1971ptak}
and \citet[4, pp.\ 421-2]{thomsen1978fubini} apply these criteria to recover
the restricted class of zero-sum games for which Karlin's result holds.
\citet{fenstad1967good} arrives at the same equilibrium existence result by a
more direct measure-theoretic route.

The condition isolated by \citet{fenstad1967good} was independently
rediscovered by \citet{marinacci1997finitely} and, in a publication-delayed
contribution, by \citet{harris2005nearly}.
The latter gave equivalent
algebraic, measure theoretic, functional analysis and finite approximability
characterizations of the condition.
The most revealing result in
\citet{harris2005nearly} shows that, after identification of strategically
equivalent strategies, the class of games satisfying this condition is
isometrically isomorphic to the class of compact games with jointly continuous
payoffs.
This is how this branch in the study of finitely additive strategies
for games returns precisely to where it began, with the
\citet{glicksberg1952further} class of compact and continuous games.

An alternative response to the failure of Fubini's theorem is to redefine the equilibrium concept itself.
\citet{flesch2017zero}
pursues this for zero-sum games, \citet[\S1]{vasquez2017essays}, for finite
player games, and \citet{flesch2021legitimate} for games with arbitrary player
sets.
Each evaluates equilibrium utilities by the upper integral while
evaluating deviation utilities to the lower integral.\footnote{For these
integrals, see e.g.\ \citet[\S 4.2]{royden1988real} or \citet[Problems
15.1-7]{billingsley}.}  This systematically severs the connection between
equilibrium utility levels and the strategic structure of the game:  in
zero-sum games, both players may receive strictly positive payoffs in
equilibrium, directly violating the antagonistic structure that zero-sum
payoff functions encode;
in the standard Bertrand model of price competition
for a homogeneous good, both firms may receive monopoly profits in
``equilibrium,'' vitiating the competitive pressure the model is designed to
capture.
In each case, the equilibrium utility levels so obtained fail to be feasible payoffs for the game.

Another change in the definition of equilibrium is in
\citet{milchtaich2023best}.  He defines a ``best-response equilibrium'' in
finitely additive strategies for a subclass of games.
To be a Milchtaich
best-response equilibrium, a vector of mixed strategies must mesh with the
utility functions in two ways.
First, for each player $i$ and each of their
actions, $a_i \in A_i$, the function $a_{-i} \mapsto u_i(a_i, a_{-i})$ be
integrable with respect to the product of the other players' finitely additive
probabilities.
That is, it requires that $v_i(a_i) = \int
u_i(a_i,a_{-i})\,d\mu_{-i}(a_{-i})$ be well-defined.\footnote{For games with
$2$ players, this is no restriction.
For games with $3$ or more players, this
requires equality of some of the iterated integrals giving payoffs and this
defines the subclass of games to which one can apply the solution concept.}
Second, defining $\overline{v}_i = \sup_{a_i \in A_i} v_i(a_i)$, each player's
mixed strategy must put zero mass on the set of $a_i$ for which $v_i(a_i) <
\overline{v}_i - \epsilon$ for each $\epsilon > 0$.
While this looks very
similar to definitions that work well for countably additive strategies that
are jointly measurable with respect to the product $\sigma$-field, as he notes
(p.\ 5), $\overline{v}_i$ ``cannot generally be interpreted as player $i$'s
\textit{equilibrium payoff}'' (emphasis in the original).
For example, it
allows the `equilibria' for the standard Bertrand price competition model, and
the counterexample to Karlin's claims in \S\ref{subsection Karlin's claims}
shows that the $\overline{v}_i$ need not sum to $0$ in zero-sum games.

The remaining works, \citet{yanovskaya1970solution}, its special case in
\citet{schervish1996fair}, and \citet[\S 5]{stinchcombe2005nash}, all deliver
finitely additive equilibria by using correlating devices to determine payoffs
at discontinuity points.
This is best understood as a species of the
endogenous sharing rule equilibria of \citet{simon1990discontinuous}, and we
will discuss them in more detail in that context.
In sum, two essential
problems attend this approach.  First, the resulting equilibria may place unit
mass on strictly dominated strategies.
Second, much as with the upper and
lower integral constructions just described, the randomization governing
payoffs may differ systematically between equilibrium play and deviations.
That said, this approach does deliver payoffs summing to $0$ in zero-sum games
and it delivers well-defined equilibrium utility levels. 

A related but distinct strand of work takes correlation itself as the object of analysis. \citet{hart1989existence} establish existence of correlated equilibria for finite games and extend the result to games with infinitely many players, finite action sets, and bounded measurable payoffs; they further remark, without proof, that an analogous result holds for finitely additive correlated equilibria in the compact Hausdorff setting (p.~24). In their formulation, the equilibrium object is a joint distribution on the product of the action spaces: correlation across players is constitutive of the solution concept, not a device for determining payoffs. The set of correlated equilibria always contains the set of Nash equilibria and in general is strictly larger; \citeauthor{hart1989existence} themselves observe that the finitely additive correlated equilibria so obtained can be ``quite unreasonable'' (p.~22).

To be sure, the equilibrium concept adopted here also departs from the classical definition: mixed strategies are finitely additive, and equilibrium is defined as immunity to pure strategy deviations. The restriction is principled: finitely additive deviations must themselves be defined by specifying the limit process by which they are reached, and the finitely approximable equilibria of \cref{thm:iterundtd} do exactly this. The three requirements on equilibrium analysis --- existence, stability, and operationalizability --- are not compromised.

The present paper delivers what the prior literature could not. The equilibrium correspondence is nonempty valued and upper hemicontinuous, and the equilibria it delivers are immune to any pure strategy deviation. Equilibrium is a product extension of individual mixed strategies, and every deviation is evaluated against the same product measure --- the contrast with \citeauthor{hart1989existence}, who define equilibrium on joint distributions and accept correlation, is exact: finite additivity is the instrument in both cases, but what it operates on determines the strategic content of the resulting equilibria. No restrictions beyond boundedness are imposed on the utility functions, foreclosing any inadvertent return to the compact and continuous setting of \citet{glicksberg1952further}. The strategic structures built into the games survive intact.

\subsection{Games without the Compactness Assumptions}

Without compactness, the existence of countably additive equilibria can no
longer be taken for granted.
The literature has addressed this challenge in
three ways: (i) by imposing boundary conditions on utility functions that push
best responses toward compact sets, an approach directly analogous to the
stability conditions for Markov chains with continuum state spaces;
(ii) by
finding conditions sufficient for the existence of approximate equilibria and
asking whether their limits are themselves equilibria;
and (iii) by
compactifying the strategy spaces so that those limits are explicitly
represented as points in a well-defined compact space.
The three approaches
overlap considerably, and each is subsumed within our framework.

\citet{meyn2009markov}, in their monumental study of the stability of Markov
chains with locally compact state spaces, identify a family of sufficient
conditions for stochastic stability.
Each is an implementation of the same
governing idea: that transition probabilities push the system back toward
compact sets whenever the state drifts too far from them.
Replacing
transition probabilities with best responses or approximate best responses
carries this idea directly into game theory.

For game and general equilibrium models, perhaps the most direct statement of
this principle is ``Of course, in order to guarantee the existence of
equilibria on a non-compact set, some kind of `boundary' assumptions (i.e.,
assumptions on utility functions outside of some compact set) are absolutely
necessary.''  \citet[p.\ 380-1, Theorem 2(iii)]{tian1992existence} gives the
boundary requirement its most direct and explicit form, and the closely
related \citet[Theorem 2]{tian1992existencegames} generalize much of the
earlier work on relaxing compactness by giving this principle a precise form.
\citet[Theorem 3.2]{tian2015existence} carries the program further still,
establishing equilibrium existence by positing a compact subset on which the
conditions sufficient --- and necessary --- for countably additive equilibrium
existence are satisfied.

When boundary conditions cannot be guaranteed, exact countably additive
equilibria may fail to exist, and the literature has turned instead to
approximate ones.
\citet[\S 3]{wald1945generalization} shows that approximate
countably additive equilibria exist for zero-sum games in which one player has
a finite set of actions and the other a countably infinite one.
\citet[\S
2]{yanovskaya1974infinite} surveys the broader literature on zero-sum games.
With $X$ and $Y$ denoting the set of countably additive mixed strategies for
the two players, she catalogues sufficient conditions for 
\[ (\ddagger)\qquad
\normsup_{x
\in X}\norminf_{y \in Y} u(x,y) \;=\;
\norminf_{y \in Y}  \normsup_{x \in X} u(x,y).
  \]
This condition is weaker than the existence of an exact countably additive
equilibrium, 
\[ (\ddagger)\quad\,  \max_{x \in X}   \min_{y \in Y} u(x,y) 
  \;=\;
\min_{y \in Y}  \max_{x \in X} u(x,y),\] 
but strong enough to guarantee the existence of countably additive
$\epsilon$-equilibria for all $\epsilon > 0$.
When $(\dagger)$ holds but
$(\ddagger)$ does not, there is a gap between having a value and having an
equilibrium.
From \cref{cor:caepsilonequilibria}, we know
that this gap does not exist in our approach.

\citet{tijs1981nash} extends the approximate equilibrium existence results for
zero-sum games to general finite player games, and in doing so draws the
connection between the two approaches into sharp relief.
His central
observation is that finite approximations to the strategy sets that nearly
deliver best responses to all strategies are themselves an instance of the
requirement that best responses push players toward compact sets --- here,
finite ones.
Equicontinuity of the utility slices, or uniform continuity of
the utility functions on uniformly bounded action spaces, guarantees the
existence of such approximations, and hence the existence of countably
additive $\epsilon$-equilibria for all $\epsilon > 0$.
What these
approximations do not supply is the limiting strategies themselves.
This lack
is something that compactification fills by enlarging the strategy spaces so
that all limits of approximating strategies are explicitly represented, giving
the limits of approximate equilibrium strategies a precise mathematical home.

\citet{young1937generalized} pioneered this approach in the calculus of
variations, embedding measurable functions from $[0, T]$ to a set of
distributions on a compact set of actions $A$ as distributions on their graphs
in the compact space of probabilities on $[0,T] \times A$ that have the uniform
distribution as the marginal on $[0, T]$.
\citet{milgrom1985distributional}
used the same compactification in game theory to study games where players
choose their mixed strategies as a measurable function of the information
contained in their signal.
\citet{lacker2015mean} and \citet{lacker2021mean}
extend the reach of this compactification further still, deploying Young
measures in the analysis of mean field games and the convergence of finite
player games to their mean field limits.

The compactification used in this paper embeds each player's set of mixed
strategies as a dense subset of the set of finitely additive probabilities
defined on the class of all subsets of actions, a space sufficiently rich to
represent all limits of approximate equilibria along nets of finite
approximations.
This represents all vectors of mutual best responses
identified in the boundary condition approach and all limits of vectors of
approximate mutual best responses identified in the approximation approach.

The finitely additive probabilities can be identified with supnorm continuous
linear functionals on the set of all bounded functions.
Restricted to vector
subspaces of functions containing the utility slices, these linear functionals
are integrals against probabilities --- countably additive ones when the
equilibria of the boundary condition approach are recovered, and finitely
additive ones when the limits of the approximate equilibria of the
approximation and compactification approaches are represented.
The present
framework thus delivers what the boundary condition, approximate equilibria,
and compactification approaches each deliver, and more.

\subsection{Games with Well-Behaved Discontinuities}

With compactness and countable additivity assumed, the question becomes what
class of discontinuities of the utility functions still encompass games of
economic interest while still being restrictive enough to guarantee countably
additive equilibrium existence.
\citet{yanovskaya1974infinite} provides a
masterful survey of the early work on zero-sum games, and
\citet{radzik1996results} provides a more detailed historical focus on games
of timing.
While many of the techniques and themes --- upper semicontinuity,
quasi-concavity, and diagonal transfer continuity --- recurred in the study of
general games, \citet{radzik1989two} showed that conditions sufficient for
equilibrium existence in zero-sum games need not remain sufficient in the
general setting.

The early literature pursued equilibrium existence by studying limits of
equilibria for games played on sequences of finite approximations to the
action spaces.
By Nash's theorem, equilibria exist for each finite
approximation; by compactness, any sequence of such equilibria has an
accumulation point.
For an accumulation point to be an equilibrium of the
limit game, two conditions must be met: the limits of the finite equilibrium
utilities must be utilities in the limit game;
and the payoffs to deviations
must not jump upward. The later literature abandoned finite approximations in
favor of direct existence arguments, identifying conditions on the
discontinuities that permit a fixed-point argument to go through directly.
We
examine these two strands directly below (\S\ref{subsubsection sequences of
finite approximations} and \S\ref{subsubsection special discontinuities}).

There is, however, a fundamental tension running through both.  The countably
additive mixed strategies on compact metric spaces are defined to converge
when their integrals against all continuous functions converge.
But this
notion of convergence is applied to discontinuous utility functions, and all
of the difficulties in this literature arise because these two are not
generally compatible.
The equilibrium utilities along a convergent sequence
of equilibrium strategies need not converge;
and even when they do, the
payoffs to deviations may jump upward in the limit, destroying the immunity to
pure strategy deviations built into the definition of equilibria.
The present
paper resolves this tension entirely by moving to finitely additive
probabilities and defining convergence by the convergence of integrals against
all bounded functions.

\subsubsection{Sequences of Finite Approximations}
\label{subsubsection sequences of finite approximations}

\citet{dasgupta1986Itheory} laid the groundwork for the study of equilibrium
existence by way of sequences of finite approximations.
Their approach imposes
three conditions on the utility functions: lower dimensionality of the
discontinuity sets, sufficient to ensure that best responses avoid them;
upper
semicontinuity of the sum of the players' utilities; and lower semicontinuity
of each player's best payoff against the choices of the other players at any
discontinuity.
Each is designed to enforce one or both of the two requirements
for accumulation points to be equilibria --- namely, that equilibrium
utilities survive the passage to the limit, and that payoffs to deviations do
not jump upward.
Together, these conditions guarantee that any accumulation
point of equilibria along any eventually dense sequence of finite
approximations are themselves equilibria.

\citet{simon1987games} generalizes \citet{dasgupta1986Itheory} in several
directions.  In essence, rather than requiring that \textit{any} accumulation
point of equilibria along \textit{any} eventually dense sequence of finite
approximations be equilibria, he seeks conditions guaranteeing only that
\textit{some} eventually dense sequence admits \textit{some} subsequence of
equilibria converging to an equilibrium of the limit game.
This permits, for
instance, utility functions that jump upward at a dominant strategy, a case
that \citet{dasgupta1986Itheory} must explicitly exclude because the inclusion
or exclusion of the dominant strategy in the finite approximation determines
whether the limit is an equilibrium.
\citet{simon1987games} also relaxes the
assumption that the sum of the utility functions be upper semicontinuous,
replacing it with what he calls complementary discontinuities --- now known as
reciprocal upper semicontinuity: whenever one player's utility jumps down at a
limit point, another's must jump up.
It is this concept, and its subsequent
generalizations, that the later literature takes as its point of departure.

\subsubsection{Special Discontinuities}
\label{subsubsection special discontinuities}

The central insight in the later literature that abandoned sequences of finite
approximations is due to \citet{reny1999existence}.
It is that equilibrium
existence can be established by a fixed-point argument provided the utility
functions have discontinuities well-behaved enough to satisfy a condition he
calls better-reply security.
This is a condition that is both easy to verify
and satisfied in a wide range of games of economic interest.
\begin{dfnn}\label{dfn well-behaved discontinuities} 
For a finite player game $\Gamma = (A_i, u_i)_{\ii}$ where each $A_i$ is a
compact metric space, the utility function $u(\cdot)$ has \textbf{well-behaved
discontinuities} if for all countably additive mixed strategies $q =
(q_i)_{\ii}$ that are \underline{not} equilibria, there is an open
neighborhood $G_q$ and a weak$^*$-continuous $q \mapsto \varphi(q) =
(\varphi_i(q_{-i})_{\ii})$ with the property that for all $q' \in G_q$,
there is at least one player $j$ such that $u_j(q'_{-j},\varphi_j(q'_{-j})) >
u_j(q')$.
\end{dfnn}

The argument that the game has an equilibrium is then a fixed-point argument.
The set of countably additive mixed strategies $\Delta^{ca}$ is compact and
convex in the weak$^*$ topology generated by integrating against continuous
functions.
If no equilibrium exists, then Definition \ref{dfn well-behaved
discontinuities} allows one to cover $\Delta^{ca}$ with open sets $G_q$;
by
compactness, a finite subcover suffices to cover the space; partitions of
unity glue the associated functions $\varphi$ into a single continuous
function on $\Delta^{ca}$;
and the \citet{glicksberg1952further}
generalization of Kakutani's theorem delivers a fixed point $q^\circ$ --- at
which, by not changing action, some player can deviate profitably.
The
contradiction establishes existence.\footnote{Several of the key advances in
this literature replaced the continuous $\varphi(\cdot)$ by a correspondence
having the fixed point property as well as the same ``someone does strictly
better'' property.
See Figure~1 (p.~451) in the \citet{reny2020nash} survey
for a map of this literature and its logical dependencies.}

At its best, this approach has much to recommend it.
Many of the conditions
are easily verified, broad enough to be applicable to wide classes of game of
interest, and still strong enough to rule out discontinuities that prevent the
existence of countably additive equilibria.
Yet these very strengths carry a
cost.  The conditions that make the fixed-point argument go through are
precisely the conditions that exclude discontinuities of significance within
the games --- the approach succeeds by assumption where it should succeed by
argument.
And Example \ref{ignore Pareto dominant eqa} is instructive about
another type of cost.
The game satisfies one of the weaker conditions in this
literature, hence does have a countably additive equilibrium, but that
equilibrium is Pareto dominated by a finitely additive equilibrium.
The
present framework identifies this as an equilibrium, one that we find to be
focal, and it is not in the set of countably additive equilibria.

The contrast with the present paper is stark. We impose no conditions on the
utility functions beyond boundedness.
What follows from our use of finitely
additive mixed strategies is a well-behaved equilibrium theory with finite
approximability built in.  And these results are available for the games that
the well-behaved discontinuities literature must exclude.

\subsection{Sharing Rule Equilibria}\label{subsec:Sharing Rule Equilibria}

The literature just surveyed restricts attention to utility functions with
special discontinuities.
By contrast, the sharing rule approach of
\citet{simon1990discontinuous} confronts wayward discontinuities directly ---
recasting them through a limit and convexification process that redefines the
utility function at the discontinuity points.

The construction proceeds in three steps.  The first step closes the graph of
the utility, yielding a payoff correspondence that is multi-valued precisely
at the discontinuity points.
The second step convexifies the range of this
closed-value correspondence, accommodating limits of arbitrarily correlated
randomization near the discontinuities.
And the third step shows that there
exists a measurable selection from the resulting convex-valued correspondence
with the property that the game with those payoffs has a countably additive
equilibrium, called a \textit{sharing rule}, or \textit{selection}
equilibrium.

The construction purchases equilibrium existence at a price --- but it offers
something in return.
It offers a finitely additive interpretation, connecting
it to the correlating device equilibria of \citet{yanovskaya1970solution},
\citet{schervish1996fair}, and \citet{stinchcombe2005nash} briefly discussed
in \S\ref{subsect fa strats}.
The price is threefold: the tie-breaking rules
it imposes may not respect independence in the random choices across the
players;
they may ignore strictly dominant strategies, yielding ``equilibria''
in strictly dominated strategies;
and they may evaluate equilibrium strategies
and deviations using different correlating devices --- the same dependence
that severs the connection to any recognizable notion of Nash equilibrium
flagged in \S\ref{subsect fa strats}.

The arbitrariness of the tie-breaking rules cuts deepest in games where they
are not a technical artifact but an essential feature of the game's
specification --- auctions foremost among them.
To change them is to analyze a
different game; and if the equilibria depend on that change, they have nothing
to say about behavior in the situation being modeled.
This is not always the
case, but it requires separate and often subtle arguments.
For example, for a
large class of auction models, \citet{jackson2005existence} show that the
choice of sharing rule values at the discontinuities does not matter.

Selection rules may ignore strictly dominated strategies, yielding equilibria
that play them.
\citet[Cor.\ 3.3.1, p.\ 347]{stinchcombe2005nash} shows that
this can be remedied by replacing sequences of finite approximations with
exhaustive nets of finite approximations --- precisely the approach adopted
here (see Definition \ref{dfn exhaustive net} below).
More subtle is the
problem of differing correlating devices. As \citet[Example 2.4 and \S2.5.2,
pp.\ 339-40]{stinchcombe2005nash} shows, the limit correlation embodied in the
convexification step can differ at deviations from what it is at putative
equilibrium strategies --- the same asymmetry between equilibrium utility
evaluation and deviation utility evaluation that, as in parts of the
literature on finitely additive probabilities has done, severs the connection
to Nash equilibrium.

The present paper sidesteps these costs entirely. By working with finitely
additive mixtures, we have no need to change the utility functions --- neither
at the discontinuities nor anywhere else.
The utility functions are taken as
given, the tie-breaking rules are not imposed, and the evaluation of
equilibrium strategies and deviations is governed by the same integral
throughout.

The finitely additive interpretation flagged above is made precise in
\citet[\S5]{stinchcombe2005nash}.
Drawing on the representation theory of
\citet[\S4]{yosida1952finitely} and recapitulated in the appendix,
arbitrary strategy spaces are embedded as a dense subset of a compact space
whose points are identified with the zero-one (Z1) probabilities.
The utility
function in the original game is then extended by the same selection
equilibrium logic, and these selections capture precisely the correlation lost
in the passage to the limit when countably additive strategies are used.
The
result is the existence of finitely additive equilibria, but ones that inherits
the weaknesses of the sharing rule construction from which they proceed.

\citet{yanovskaya1970solution} and \citet{schervish1996fair} pursue the same
embedding for zero-sum games, but with further restrictions:  the latter works
in a rather difficult-to-identify subset of the zero-sum games, while the
former treats all zero-sum games.
They both use mid-point rules to determine
limit utilities at discontinuities, and they both assign their single choice
of payoff to all discontinuity points independent of the nearby payoffs.
The
result is a class of equilibria related to sharing rule equilibria defined by
a particular tie-breaking convention rather than by the strategic logic of the
game itself.

It is here that the present paper's results find their sharpest expression.
The finitely approximable equilibria of \cref{thm:iterundtd} admit a
sharing rule interpretation, but without play of strictly dominated
strategies, without arbitrary tie-breaking rules, and without the asymmetry
between the evaluation of equilibrium strategies and deviations that the
sharing rule construction allows.

\section{Equilibria and their Properties}\label{sec:eqandproperties}

This section starts with the class of games under study and then turns to the
basic definitions and properties of the class of finitely additive
probabilities that we use to model mixed strategies.
As emphasized in the
introduction, our coverage emphasizes
the compactness and finite approximability properties of these probabilities.
It then defines finitely additive equilibria as those that are immune to pure
strategy deviations to any $b_i \in A_i$ for any $\ii$, and states the three
main results:  the existence of equilibria;
the continuity properties of the
equilibrium correspondence;  and the same pair of results for a refinement
that deletes the iteratively weakly dominated strategies.
We give sketches of
some of the proofs in the text, details are relegated to the appendix.

\subsection{Games with Bounded Utilities}

The following defines the class of games studied here.
\begin{dfnn}

$\Gamma = (A_i,u_i)_{\ii}$ is a \textbf{game with bounded utilities} if

\begin{enumerate}[(1), itemsep=1pt, topsep=3pt,leftmargin=1em]
  \item  $I$ is a non-empty set of agents;
  \item  for each $\ii$, $A_i$ is a non-empty set of actions;
and 

  \item  there is a $B > 0$ such that for each $\ii$, each $u_i:A \rar [-B,
+B]$, with $A \coloneqq \bigtimes_{j\in I} A_j$, is a bounded von Neumann-Morgenstern
utility function.
\end{enumerate}
\end{dfnn}

There are no assumptions on the cardinality of the set of players.
There are
no topological or measure theoretic assumptions on the sets of actions.
There
are no assumptions on the utility functions except boundedness: the uniformity
of the bound across agents is without loss of generality;
and some bound is
necessary to preclude phenomena similar to the St.\ Petersburg paradox.

\subsection{Probabilities and their Properties}

The starting point is the set of total probabilities.
When we turn to games,
the set $X$ in the following will be $A = \ppi A_i$.
\begin{dfnn}
For $X$ a non-empty set and $\mcalX$ denoting the class of all subsets of $X$,
a \textbf{total probability on $X$} is a function $p:\mcalX \rar [0, 1]$ that
satisfies $p(X) = 1$ and $p(B_1 \cup B_2) = p(B_1) + p(B_2)$ for all disjoint
$B_1,B_2 \in \mcalX$.
A probability is \textbf{Z1} or \textbf{zero-one} if
$p(B) = 0$ or $p(B) = 1$ for all $B$.
The set of probabilities on the class
of all subsets of $X$ is denoted by $\Delta(X)$ or $\Delta$ when $X$ is clear
from context.
\end{dfnn}

By induction, if $\{B_n: n = 1, \ldots , N\}$ is a finite collection of
disjoint sets, a total probability must satisfy $p(\cup_{n = 1}^N B_n) =
\sum_{n = 1}^N p(B_n)$.

\subsubsection{Convergence and Compactness}
  
Every total probability on $X$ can be identified with a point $p$ in the
product space $[0, 1]^{\mcalX}$.
Being the product of compact spaces, $[0,
1]^{\mcalX}$ is compact in the product topology (by Tychonov's theorem).
In
this topology, convergence is defined by $p^\alpha \rar p$ if $p^\alpha(B)
\rar p(B)$ for all sets $B$.
Since finite sums are continuous in their
arguments, $\Delta$ is a closed, hence compact set.

\subsubsection{Integrals}

We give the bounded functions on $X$ the sup norm metric, $\|f-g\| = \sup_{x
\in X}|f(x) - g(x)|$.
The simple functions on $X$ are sup norm dense in the
set of bounded functions $f$.\footnote{Use the classic Lebesgue approximations
$f_n(x) = \sum_{k=-n2^n}^{+n2^n} \frac{k}{2^n} 1_{E_{n,k}}(x)$ where $E_{n,k}
= \{x : \frac{k}{2^n} < f(x) \leq \frac{k+1}{2^n}\}$ to see this.}  The integral
of any simple function is well-defined.
The mapping from the class of simple
functions to their integral has unit Lipschitz constant.
The integral for
bounded functions is defined as the unique Lipschitz continuous extension from
the dense set of simple functions.
Further, $p^\alpha \rar p$ if and only if
$\int f\,dp^\alpha \rar \int f\,dp$ for all bounded functions $f$.

It is perhaps worth emphasizing the simplification of the measure theory that
comes from all sets and functions being measurable and the probabilities being
total.

\subsubsection{Extensions}

The compactness of $\Delta$ is equivalent to the property that every
collection of closed subsets having the finite intersection property has a
non-empty intersection, which delivers the following.\footnote{Details of
proofs not in the text are in the appendix.}  

\begin{lemma}\label{lemma about extensions}
If $\mcalX^\circ$ is a field\,\footnotemark~ of subsets of $X$ and
$q:\mcalX^\circ \rar [0, 1]$ satisfies $q(B_1 \cup B_2) = q(B_1) + q(B_2)$ for
disjoint $B_1, B_2 \in \mcalX^\circ$, then the set of total probabilities that
extend $q$ from $\mcalX^\circ$ to $\mcalX$ is a non-empty, compact and convex
set of probabilities.
\end{lemma}\footnotetext{A class of subsets of $X$ is a field if it contains
$X$, is closed under complementation and finite unions and intersections.}

\subsubsection{Finitely Supported Approximations}

For $F$ a finite subset of $X$, $\Delta(F)$ is the set of probabilities
satisfying $p(F) = 1$.
To talk about limits of approximating finite games in
such a fashion that we can guarantee that every action of every player is
eventually included requires the following generalization of sequences.
\begin{dfnn}
A pair $(D, \succsim)$ is a \textbf{directed set} if $D$ is nonempty and
$\succsim$ is a transitive binary relation on $D$ satisfying: $\alpha \succsim
\alpha$ for all $\alpha \in D$, and for all $\alpha, \beta \in D$, there
exists $\gamma \in D$ with $\gamma \succsim \alpha$ and $\gamma \succsim
\beta$.
A \textbf{net} in a set $X$ is a mapping $\alpha \mapsto x^\alpha \in
X$ from a directed set $D$ to $X$.
\end{dfnn}

\noindent Sequences arise as the special case where the directed set is
$(\integers, \geq)$.
We will need nets of finite approximations to a game as
well as nets of equilibria for those approximate games.
\begin{dfnn}\label{dfn exhaustive net}
A \textbf{net of finite approximations to a set $X$} is a mapping $\alpha
\mapsto F^\alpha$ from a directed set $(D, \succsim)$ to the class of finite
subsets of $X$.
We write $F^\alpha \uparrow \infty$ if the net is
\textbf{exhaustive for $X$}, that is, if for all finite $F \subset X$, there
exists an $\alpha \in D$ such that for all $\beta \succsim \alpha$, $F \subset
F_\beta$.
\end{dfnn}

\noindent
Sequences of finite sets can exhaust countable sets.  The right choice of
indexing set shows that nets of finite sets can exhaust any set:  let $D$
denote the class of finite subsets of a set $X$;
for $F, F' \in D$, define $F
\succsim F'$ if $F \supset F'$;
taking the mapping from $D$ to the finite
sets to be the identity mapping, we have, for all finite $F \subset X$, there
exists an $\alpha \in D$, namely $\alpha = F$, such that for all $\beta
\succsim \alpha$, $F \subset F_\beta$.

From \citet[Cor.\ 1.2]{stinchcombe2023direct}, we have the following
informative characterization of exhaustive nets.
\begin{lemma}\label{lemma characterizing exhaustive nets}
A net $\alpha \mapsto F^\alpha$ of finite subsets of $X$ is exhaustive if and
only if for all total probabilities $p \in \Delta(X)$, there is a net $\alpha
\mapsto p^\alpha \in \Delta(F^\alpha)$ with $p^\alpha \rar p$.
\end{lemma}

\noindent Said differently, $\alpha \mapsto F^\alpha$ is exhaustive if and
only if every $p \in \Delta(X)$ is an accumulation point of the sets
$\Delta(F^\alpha)$.

\begin{dfnn}
A total probability $\mu \in \Delta$ is a \textbf{limit point} of the net
$\alpha \mapsto \mu_\alpha$ of total probabilities if for all sets $B \subset
A$ and all $\epsilon > 0$, there exists an $\alpha \in A$ such that for all
$\beta \succsim \alpha$, $|\mu_\beta(B) - \mu(B)|
< \epsilon$, and $\mu$ is an
\textbf{accumulation point} of the net $\alpha \mapsto \mu_\alpha$ if for all
sets $B \subset A$, all $\epsilon > 0$, and all $\alpha$, there exists $\beta
\succsim \alpha$ such that $|\mu_\beta(B) - \mu(B)|
< \epsilon$.  
\end{dfnn}

We are now in a position to discuss finitely additive equilibria.

\subsection{Equilibrium Existence}

Nash equilibria require independent randomization by the players.  The
existence of the following extensions is guaranteed by Lemma \ref{lemma about
extensions}.

\begin{dfnn}\label{dfn independent extensions}
A probability $\widehat \mu \in \Delta$ is a \textbf{independent extension} of
a vector 
$(\mu_i)_{\ii} \in \ppi \Delta_i$ if for all finite $I_F \subset I$, for all
$B_j \subset A_j$ for $j \in I_F$, 
\begin{equation}
\widehat{\mu}\Bigl( \textstyle{\proj^{-1}}\displaystyle\bigl(\bigtimes_{j \in I_F} B_j\bigr)\Bigr) = 
   \prod_{j \in I_F} \mu_j(B_j).
\end{equation}
\end{dfnn}

We use the following game-theoretic notation, for $a \in A$, $\ii$, and $b_i
\in A_i$, the point $a \backslash b_i \in A$ is defined by
$\proj_j(a\backslash b_i) = a_j$ for $j \neq i$ and $\proj_i(a\backslash b_i)
= b_i$.
And we extend this notation to mixtures over $A$, for $\mu$ a
probability on $A$, $\ii$ and $b_i \in A_i$, $\mu \backslash b_i$ is the image
measure of $\mu$ under the mapping $a \mapsto a\backslash b_i$ from $A$ to
$A$.

\begin{dfnn}\label{dfn Nash eqm}
An independent extension $\mu^*$ of $(\mu^*_i)_{\ii}$ is a \textbf{Nash
equilibrium} if for all $\ii$ and all $b_i \in A_i$, $u_i(\mu^*) \geq
u_i(\mu^* \backslash b_i)$.
\end{dfnn}

\subsubsection{Existence}

Equilibria exist.

\begin{thmm}\label{thm:equilibriaexist}
For any game with bounded utilities, an equilibrium exists.
\end{thmm}

The proof for the general case is in the appendix, here we sketch the argument
for games with finite player sets.\footnote{\citet{cerreia2022equilibria}
provides a finitely additive equilibrium existence result for a subset of the
nonatomic population games with a finitely additive population measure.
\cref{thm:equilibriaexist} removes the restrictions on the utility
functions and information structures used in that work.}  For each $\ii$ let
$F_i$ be a finite subset of $A_i$.
For $\epsilon > 0$, let
$Eq^\epsilon((F_i,u_i)_{\ii})$ denote the set of $\epsilon$-equilibria for the
finite game played with actions sets $(F_i)_{\ii}$ with the utility functions
restricted to $\ppi F_i$.
For finite $F = \ppi F_i$ and $\epsilon > 0$, let
$E(F,\epsilon)$ denote the closure of the set of
$Eq^\epsilon((F'_i,u_i)_{\ii})$ with $F_i \subset F'_i$ for each $\ii$.
The
class of closed sets $E(F,\epsilon)$ has the finite intersection property, and
since $\Delta$ is compact, it therefore has non-empty intersection.
Any
$\mu^* \in \bigcap E(F,\epsilon)$ is an equilibrium where the intersection is
taken over finite product sets $F \subset A$ and $\epsilon > 0$.

\subsubsection{Finitely Approximable Equilibria}

In the previous argument, we take the intersection over all finite sets $F =
\ppi F_i$.
This means that any point in the intersection is finitely
approximable.
The zero-sum game in \S\ref{subsection largest integer} shows
that not all finitely additive equilibria are finitely approximable, but we do
have the following (much) weaker statement.
Recall that a Z1 probability is
one for which $p(B)$ is either equal to $0$ or $1$ for all sets $B$.

\begin{corro}\label{all Z1s are lof equilibria}
If $\mu$ is a {\rm Z1} finitely additive equilibrium for a finite player game
$\Gamma = (A_i,u_i)_{\ii}$, then $\mu$ is a limit point of a net of $\alpha
\mapsto \mu^\alpha$ of $\epsilon^\alpha$-equilibria for a net of finite games
$\alpha \mapsto (F_i^\alpha,u_i)_{\ii}$ with $F_i^\alpha \uparrow \infty$ and
$\epsilon^\alpha \rar 0$.
\end{corro} 

In outline, the proof begins with some straightforward observations and ends
with a subtle one, and again, the full argument is in the appendix.
First, a
slight sharpening of Lemma \ref{lemma characterizing exhaustive nets} shows
that for any net of finite approximations $\alpha \mapsto \ppi F_i^\alpha
\uparrow \infty$, every Z1 is the limit of a net $\alpha \mapsto \eta^\alpha$
of Z1's on $\ppi F_i^\alpha$.
Second, each player's net of payoffs,
$v_i^\alpha := \int u_i(a)\, d\eta^\alpha(a)$ converges to their equilibrium
payoffs, $v_i := \int u_i(a)\,d\mu(a)$.
Third, each possible deviation $b_i
\in A_i$ for $i$ is eventually in all of the $F_i^\alpha$, and the payoff to
that deviation, $\int u_i(a \backslash b_i)\,d\eta^\alpha(a)$, converges to
$\int u_i(a \backslash b_i)\,d\mu(a)$.
The fourth, and somewhat subtle step
is to show that one can remove from the $F_i^\alpha$ the non-constant nets of
points converging to payoffs strictly larger than $v_i$.
And by the previous
step, this removes none of the $b_i$.

\subsubsection{Compact and Continuous Games}

Total probabilities have a dual representation as continuous linear
functionals on the vector space of all bounded functions.
By restricting the
linear functionals to vector subspaces containing the utility functions, one
arrives at utility equivalent sets of equilibria.
\S\ref{subsubsection special discontinuities} examined this issue in more detail, but here we offer the
following immediate consequence of \cref{thm:equilibriaexist} and the
Riesz representation theorem.\footnote{For compact Hausdorff spaces, the
continuous dual of the space of continuous functions is the set of countably
additive measures, see e.g.\ \citet[Theorem IV.6.3, p.\
265]{dunford1958linear}.}

\begin{corro}[Generalized Glicksberg]\label{cor:generalizedGlicksberg}
For $\Gamma = (A_i, u_i)_{\ii}$, if each $A_i$ is a compact Hausdorff space
and each $u_i:A \rar [-B, +B]$ is continuous in the product topology on $A$,
then for any finitely additive equilibrium $\mu^*$, the unique 
countably additive $p^*$ satisfying $\int_A f(a)\,dp^*(a)=\int_A f(a)\,d\mu^*(a)$ for all continuous
$f$ is an equilibrium for $\Gamma$.
\end{corro}

\citet{glicksberg1952further} proved equilibrium existence for this class of
games when $I$ is finite.
The generality offered by the use of compact
Hausdorff spaces of actions rather than compact metric spaces is illusory for
such games.\footnote{\citet{harris2005nearly} showed that finite player games
with compact Hausdorff spaces of actions and jointly continuous utilities are,
after identification of strategically equivalent strategies, games with
compact \textit{metric} spaces with jointly continuous utilities.
We
conjecture that the same is true when $I$ is infinite.}

\subsection{Properties of the Set of Equilibria}

Finitely additive probabilities are often understood as having ``various
peculiar properties'' (e.g.\ \citet[p. 56]{yosida1952finitely}).
One might
worry that the use of such probabilities would deliver an equilibrium theory
that is not recognizable.
The next section analyzes several examples with a
view to understanding how to work with finitely additive equilibria while the
following result offers some theoretical reassurance.

To focus on the dependence on the utility function, we denote the game $\Gamma
= (A_i,u_i)_{\ii}$ as $\Gamma(u)$ where $u:A \rar [-B, +B]^I$, and for each
game $\Gamma(u)$, $Eq(u) \subset \Delta$ denotes the set of equilibria and $u
\mapsto Eq(u)$ is the equilibrium correspondence.
We define the convergence
of utility functions by $u^\alpha \rar u$ if (and only if) for all $\ii$, the
supnorm distance, $\|u_i^\alpha(\cdot) - u_i(\cdot)\|$, converges to $0$.
The
proof of the following assertion is a minor variant on the textbook arguments
for finite games.

\begin{thmm}\label{thm:closedeqmset}
For any game with bounded payoffs, the equilibrium correspondence is
non-empty valued, closed valued, and upper hemicontinuous.
\end{thmm}

There are many games for which there are countably additive
$\epsilon$-equilibria for all positive $\epsilon$, but for which there does
not exist a countably additive equilibrium.
By contrast, the logic of the
proof of \cref{thm:closedeqmset} immediately delivers the following.

\begin{corro}\label{cor:caepsilonequilibria}
If for every $\epsilon > 0$, $Ca(\epsilon)$, the set of countably additive
$\epsilon$-equilibria, is non-empty, then $\displaystyle\bigcap \{\closure(Ca(\epsilon) :
\epsilon > 0\}$ contains an equilibrium.  
\end{corro}

There are many examples of games for which there are approximate countably
additive equilibria but no countably additive equilibrium.\footnote{Bertrand
price competition for a homogenous good but with different marginal costs is
such a game.  \citet{glicksberg1950minimax} shows that zero-sum games on
compact metric spaces having upper semi-continuous payoffs for player $1$,
hence lower semi-continuous payoffs for player $2$, can belong to this class.
Other examples include \citet{laraki2005continuous}, which gives Markov
perfect $\epsilon$-equilibria for every $\epsilon > 0$ for general timing
games, \citet[Theorem 2.9, p.\ 
280]{barelli2014competition}, which gives
$\epsilon$-equilibria for every $\epsilon > 0$ in games modeling competitions
for a majority.}

\subsection{Equilibria in Iteratively Undominated Strategies}

We analyze the examples in the next section by replacing each player's set of
actions by an exhaustive net of finite approximations and studying the limits
of equilibria after iterated deletion of weakly dominated strategies in the
net of approximating finite games.
\cref{thm:iterundtd} below shows that
this process always delivers a non-empty, closed set of equilibria with a
well-behaved equilibrium correspondence.

Our use of finitely supported probabilities in the following is neither usual
nor without loss of generality, but it is appropriate for our exhaustive
finite nets approach to finitely additive equilibria.

\begin{dfnn}
In a game $\Gamma = (A_i,u_i)_{\ii}$, an action $c_i \in A_i$ is
\textbf{weakly dominated for $i$} if there is a finitely supported probability
$q_i$ on $A_i$ such that for all $a \in A$, $u_i(a \backslash q_i) \geq u_i(a
\backslash c_i)$ and the inequality is strict for at least one $a$.
\end{dfnn}

Given weak dominance, iterated weak dominance is defined as usual.
\begin{dfnn} 
For a game $\Gamma^0 = (A^0_i,u^0_i)_{\ii}$, for   each $\ii$, let $D^0_i$
denote the set of weakly dominated strategies in $\Gamma^0$, let $A^1_i$
denote $A^0_i \setminus D^0_i$, and define $\Gamma^1 = (A^1_i,u^1_i)_{\ii}$ by
restricting each $u_i$ to $\times_{j \in I} A^1_j$.
Iteratively 
apply this:  given a game $\Gamma^n = (A^n_i,u^n_i)_{\ii}$, let
$D^n_i$ denote $i$'s weakly dominated strategies, let $A^{n+1}_i = A^n_i
\setminus D^n_i$;
and define $\Gamma^{n+1} = (A^{n+1}_i,u^{n+1}_i)_{\ii}$ by
restricting each $u_i$ to $\times_{j \in I} A^1_j$.
Finally, let
$A^{\infty}_i = \cap_{\nn} A^n_i$ and define the \textbf{game in iteratively
undominated strategies} as $\Gamma^{\infty} = (A^{\infty}_i,
u^{\infty}_i)_{\ii}$ by restricting each $u_i$ to $\times_{j \in I}
A^{\infty}_j$.
\end{dfnn}

For a game with bounded payoffs $\Gamma(u) = (A_i,u_i)_{\ii}$, let
$\eqfin(\Gamma(u))$ denote the set of limits of equilibria for the finite
games $(F_i^\alpha,u_i)_{\ii}$ along any exhaustive net $\alpha \mapsto \ppi
F_i^\alpha$ of finite approximations to $\ppi A_i$, and let $\eqiwu(\Gamma(u))
\subset \eqfin(\Gamma(u))$ denote the set of limits of iteratively weakly
undominated equilibria for those finite games.
Note that along any exhaustive
net of finite approximations, any iteratively weakly dominated strategy is
eventually excluded from every finite approximation to the game along the net.

\begin{thmm}\label{thm:iterundtd}
The correspondences $u \mapsto \eqfin(\Gamma(u))$ and $u \mapsto
\eqiwu(\Gamma(u))$ are nonempty valued, closed valued, and upper
hemicontinuous, and no equilibrium in $\eqiwu(\Gamma(u))$ puts positive mass
on the set of iteratively weakly dominated strategies.
\end{thmm}

The existence proof directly parallels the proof of \cref{thm:equilibriaexist} and the rest of the arguments directly parallel the
arguments for \cref{thm:closedeqmset}.
The next section mostly analyzes
the set $\eqiwu(\Gamma(u))$, which is generally a proper subset of the
finitely approximable equilibria, which is, in turn, generally a proper subset
of the finitely additive equilibria.

\subsubsection{Failures of Fubini are Moot}

We here record the result that tells us that worries about failures of
Fubini's theorem are moot when we use limits of finitely supported equilibria.
The proof, simple and omitted, depends only on the continuity of finite
multiplications as this implies that the limit of a net of product extensions
is itself a product extension.

\begin{lemma}\label[lemma]{lem:fubiniismoot}
If $\mu \in \Delta(A)$ is the limit of any net of finitely supported product
probabilities $\mu^\alpha = \ppi \mu^\alpha_i$ on $A$ and $\mu_i = \lim_\alpha
\mu^\alpha_i$, then $\mu$ is a product extension of $(\mu_i)_{\ii}$.
\end{lemma}

\subsubsection{Weak Dominance}

The simplest class of infinite games are the ones with compact metric spaces
of actions and jointly continuous utility functions.
There are examples of
such games for which all countably additive equilibria put mass one on the set
of weakly dominated strategies.
\cref{thm:iterundtd} tells us that there
is a well-behaved class of finitely additive equilibria that never put mass on
the weakly dominated strategies.
The following is a simplification of \citet[Ex.\ 2.1, p.\
1428]{simon1995equilibrium} in which one can see what is at work in the
contrast.

\begin{exx}\label{S and S weakly dominated}

With $I = \{1, 2\}$ and $A_i = [0, 1]$, suppose that the jointly continuous
utility functions $u_i(\cdot, \cdot)$ take values in $[0, 1]$, are symmetric,
$u_1(a_1,a_2) = u_2(a_2,a_1)$, and have the following properties:  if either
player plays $0$, then both receive a utility of $0$;
for each $a^\circ_j >
0$, $a_i \mapsto u_i(a_i,a_j^0)$ is strictly increasing on $[0, a_j^\circ/2]$, 
and strictly decreasing on $[a_j^\circ/2, 1]$.
\end{exx}\smallskip

\noindent
By induction, best responses of each player to whatever the other player is
doing must be a subset of $[0, 1/2^n]$ for all $\nn$.
The unique countably
additive strategies with this property play the weakly dominated strategy $0$
with probability $1$.
By induction again, the set of weakly undominated
strategies are, for both players, a subset of $(0, 1/2^n]$ for all $\nn$.  Any
finitely additive equilibrium puts mass $1$ just over $1$, that is, it puts
mass $1$ on each of the sets $(0, \epsilon) \times (0, \epsilon)$, hence it
puts mass zero on the iteratively weakly dominated strategies.\footnote{\S
\ref{subsection just under} systematically covers the three representations of 
``just under'' and ``just over'' that we use.}

From \citet{simon1995equilibrium}, there are countably additive equilibria for
compact and continuous games that are \textit{limit admissible}, that is, they
only put mass on limits of weakly undominated strategies.  For the equilibria
in 
$\eqiwu(\Gamma)$ for compact and continuous game $\Gamma$, the countably
additive equilibria identified are necessarily limit admissible.

\section{Two-Person Games on the Line}\label{sec:gamesontheline}

The examples in this section put the framework of \cref{sec:eqandproperties} to work.
All five games have action sets that are subsets of the
real line; only one has a countably additive equilibrium.
Each is chosen to
illuminate a distinct feature of the finitely additive approach --- and to
contrast it with what the prior literature could and could not deliver.

The first three examples --- the asymmetric location game, the sharing rule
game, and the Sion and Wolfe game --- are analyzed by iterative deletion of
weakly dominated strategies on exhaustive hyperfinite approximations, with
sketches of how to recast the arguments using nets.
In each case, the limiting
process delivers products of finitely additive equilibrium distributions with
well-defined integrals, and the multiplicity of product extensions is not at
issue.

The last two examples are of a different character. They show that the
set of finitely additive equilibria can be a strict superset of the set of
limits of approximate equilibria along exhaustive nets, and that what appear
to be mutual best responses in zero-sum games can deliver utilities that do not
sum to zero --- precisely the pathology that the iterated integral approach of
the prior literature failed to diagnose.

A recurring theme across the first three games is the representation of a
player choosing a number ``just under'' a given value.
In the asymmetric
location game, one player's equilibrium strategy concentrates mass just below
the action $a = 0.8$.
In the Sion and Wolfe game, both players may choose just
under the action $a = 0.5$, and the question becomes which player can approach
it more closely.
\S\ref{subsection just under} examines systematically the
three representations of this idea --- exhaustive nets, hyperfinite sets, and
finitely additive limits --- and shows that all three are equivalent.

\subsection{An Asymmetric Location Game}

The asymmetric location game is the simplest example in which the payoff
discontinuity at a single point precludes a countably additive equilibrium
while the finitely additive framework delivers one immediately.
Consumers are
distributed uniformly on $[0,1]$. Licensing restrictions confine player $1$
to locations in $A_1 = [0,0.8]$ and player $2$ to locations in $A_2 =
[0.8,1]$.
Given choices $x < y$, each consumer patronizes the nearest
location --- the consumer at the midpoint $\half x + \half y$ is indifferent,
and their choice has no effect on payoffs.
The discontinuity arises at $x = y
= 0.8$: the consumers then view the two locations as perfect substitutes, with
$\half$ patronizing each player.

\subsubsection{The Equilibria}

The payoff discontinuity at $(x,y) = (0.8,0.8)$  precludes a countably additive equilibrium --- but there are $\epsilon$-equilibria for every $\epsilon > 0$.
By \cref{cor:caepsilonequilibria}, finitely additive limits of these $\epsilon$-equilibria exist and are themselves equilibria.
In fact, more is true: the equilibria admit a complete characterization in terms of the mass each player concentrates near the discontinuity point $0.8$.

\begin{lemma}\label{Simon-Zame lemma}
Any product extension of $(\mu_1^*,\mu_2^*)$ is an equilibrium for the asymmetric 
location game if and only if for all $\epsilon > 0$, 
\begin{equation}
(\dagger) \ \ \mu_1^*((0.8-\epsilon,0.8)) = 1 \ \mbox{and} \ (\ddagger)\ \ 
\mu_2^*([0.8,0.8+\epsilon)) = 1 .
\end{equation}
All finitely additive equilibria are limits of $\epsilon$-equilibria along any
exhaustive net, and if $(\mu_1^*,\mu_2^*)$ is an equilibrium in iteratively
weakly undominated strategies, then $\mu^*_1$ is a {\rm Z1} satisfying
$(\dagger)$, and $\mu^*_2$ is the countably additive point mass on $0.8$.
\end{lemma}

Condition $(\dagger)$ requires player $1$ to concentrate mass arbitrarily close to but strictly below $0.8$ --- the finitely additive representation of choosing ``just under'' $0.8$.
Condition $(\ddagger)$ requires player $2$ to concentrate mass arbitrarily close to and including $0.8$ from above.
The equilibrium is thus asymmetric despite the symmetric location of the discontinuity: player $1$ approaches $0.8$ from below, player $2$ from above, and the countably additive limit --- which would place both players at $0.8$ with probability $1$ --- loses precisely this information.

\subsubsection{The Sharing Rule Equilibrium Interpretation}

Along any exhaustive net or sequence of finite approximations that becomes
dense in the strategy spaces $A_1 = [0, 0.8]$ and $A_2 = [0.8, 1]$,
approximate equilibria put mass going to $1$ the intervals $(0.8-\epsilon,
0.8)$ and $[0.8, 0.8+\epsilon)$ for all $\epsilon > 0$.
The equilibrium payoffs converge to $(0.8, 0.2)$.  If one insists that the limits of
the equilibria must be countably additive, the only possibility is for both
players to be playing $0.8$ with probability $1$.
But if that is how one
describes how the players behave, then the payoffs are $(0.5, 0.5)$, and we
are not at an equilibrium.

The difficulty is precisely what \S\ref{subsec:Sharing Rule Equilibria}
identified: the passage to a countably additive limit destroys the equilibrium
property, and restoring it while using the countably additive model of mixed
strategies requires redefining the utility function at the discontinuity.
That is to say, analyzing a different game. The utility functions are
discontinuous at the single point $\ba^\circ = (0.8,0.8)$.
Let
$\Phi(\ba^\circ)$ denote the closure of the set of limits of payoff vectors
$(u_1(\ba^n),u_2(\ba^n))$ as $\ba^n \rar \ba^\circ$.
\citet{simon1990discontinuous} regard the payoffs at the discontinuity as
``only partially determined'': whenever the economic nature of the problem
leads to indeterminacies, they propose that the \textit{sharing rule} --- the
choice of a point out of $\Phi(\ba^\circ)$ --- be ``determined endogenously''
--- specifically, by selecting a point from the convex hull of the limits of
utility vectors along sequences of equilibria on finite approximations to the
game.
An endogenous sharing rule equilibrium is then a Nash equilibrium for a
game with utility functions taking some value in $\conv(\Phi(\ba^\circ))$.
But
the existence of such limit payoffs does not guarantee that selection
equilibria are mutual best responses --- as  \S\ref{subsec:Sharing Rule
Equilibria} discussed and we now explicitly show, they can put unit mass on
strictly dominated strategies.

\subsection{Sharing Rule Equilibria Playing Strictly Dominated Strategies}

\citet[Example~2.2, p.~337]{stinchcombe2005nash} makes good on the warning 
just issued: he presents an example that delivers an endogenous sharing rule ``equilibrium'' 
that puts mass $1$ on a strictly dominated strategy, and is therefore not a 
Nash equilibrium.
\citet[Theorem~3.3, p.~344]{stinchcombe2005nash} shows that 
the remedy is precisely the one adopted here: exhaustive nets of finite 
approximations, rather than sequences of finite sets that become dense in the 
metric topology --- a technique whose roots go back to 
\citet{simon1995equilibrium}.
\citet[Dfn.~2.16, p.~85]{bich2017existence} 
strengthen the definition of a pure strategy endogenous sharing rule equilibrium 
in a fashion that rules out some of the other problems with the sharing rule 
construction --- but still falls short: it fails to rule out equilibria that 
play strictly dominated strategies.

The following adaptation of the \citet{stinchcombe2005nash} example 
illuminates precisely why.
The game $\Gamma = (A_i,u_i)_{i=1,2}$ is 
described as follows: the action sets are $A_1 = A_2 = [0,1]$;
player $2$'s 
utility function is $u_2(a_1,a_2) = a_2$ if $a_2 > 0$, and $u_2(a_1,0) = 2$ 
if $a_2 = 0$ so that every $a_2 > 0$ is strictly dominated by $a_2 = 0$;
player $1$'s utility function has two parts, for $a_1 > 0$, $u_1(a_1,a_2) = 
a_1$, and for $a_1 = 0$,
\begin{equation}
u_1(0,a_2) = \begin{cases}
 2 & \ \mbox{if}\ a_2 = 0 \\
 0 & \ \mbox{if}\ a_2 > 0 .
\end{cases}
\end{equation}
This game is dominance solvable and the unique equilibrium is 
$(a_1^*,a_2^*) = (0,0)$ with equilibrium utility levels $u^* = (2,2)$.

Specialized to this game, \citet[Dfn.~2.16 and 2.17, p.~85--86]{bich2017existence} 
define a vector of actions $(b_1^\dagger,b_2^\dagger) \in A_1 \times A_2$ to be 
a \textit{pure sharing rule equilibrium} if $(b_1^\dagger,b_2^\dagger)$ is a 
pure strategy Nash equilibrium for an auxiliary game $\widetilde{\Gamma} = 
(A_i,g_i)_{i=1,2}$ where for all $(b_1,b_2) \in A_1 \times A_2$, the auxiliary 
utility function $g(b_1,b_2) \in \reals^2$ belongs to the closure of the graph 
of the utility function $u$ --- unlike \citet{simon1990discontinuous}, without 
convexifying the set of limit payoffs at the discontinuities.
For the game 
described above, the unique continuous selection from the closure of the graph 
of the utility functions is $g_1(a_1,a_2) = a_1$ and $g_2(a_1,a_2) = a_2$, 
and the unique Nash equilibrium of the auxiliary game is 
$(b_1^\dagger,b_2^\dagger) = (1,1)$ with utility levels $u^\dagger = (1,1)$ --- 
a strictly dominated strategy profile delivering utilities strictly below the 
unique equilibrium $u^* = (2,2)$.

The inadequacy of metric density is on full 
display in the following observations.
Let $n \mapsto (F_1^n,F_2^n)$ be a 
sequence of finite sets for which $d(F_1^n \times F_2^n,[0,1] \times [0,1]) 
\rar 0$.
\begin{enumerate}[{Obs }1.]
\item If $F_2^n$ fails to contain the dominant strategy $a_2 = 0$, then the 
unique limit of approximate equilibrium play is the sharing rule 
``equilibrium'' $(b_1^\dagger,b_2^\dagger) = (1,1)$, which plays a strictly 
dominated strategy.
\item If $n \mapsto F_2^n$ contains $0$ for all large $n$ but $0 \notin 
F_1^n$, then the unique limit of approximate equilibria is the endogenous 
sharing rule ``equilibrium'' $(c_1^\dagger,c_2^\dagger) = (1,0)$ with payoffs 
$(1,2)$, which involves a strictly dominated response for player $1$.
\item If $\alpha \mapsto (F_1^\alpha,F_2^\alpha)$ is any exhaustive net of 
finite approximations to $[0,1] \times [0,1]$, then the unique limit of 
approximate equilibria is the unique Nash equilibrium for the game, 
$(a_1^*,a_2^*) = (0,0)$ with equilibrium payoffs $(2,2)$.
\end{enumerate}

The contrast between Obs~1--2 and Obs~3 identifies the essential problem 
precisely.
The sharing rule construction is adapted to arguments that apply 
to all approximating \textit{sequences} of finite subsets that become dense 
in the metric topology --- but metric density has no bearing on whether 
dominant strategies or strict best responses are included when the utility 
functions are not continuous.
Exhaustive nets, by contrast, guarantee that 
every action is eventually included in every finite approximation, and it is 
this guarantee --- not metric density --- that delivers the correct equilibrium.

\subsection{The Sion and Wolfe Game}\label{subsection Sion and Wolfe}

\citet{sion1957game} give an asymmetric, two battlefield, Colonel Blotto game
where both players have one unit of force to allocate between the
battlefields.
The player allocating the larger/smaller force to a battlefield
wins/loses, and equal force allocations lead to ties.
The asymmetry is that
player $2$ starts with an advantage of $0.5$ immobile units of force already
present in the second battlefield.
The payoffs are additive, with $+1$ for
each battlefield won, $-1$ for each one lost, and $0$ for ties.

Let $x$ and $(1-x)$ denote player $1$'s allocations to the first and second
battle fields respectively, and let $y$ and $(1-y)$ denote the corresponding
allocations for player $2$.
The payoffs are: 
\begin{align}
v_1(x,y)\;=\; & \sgn(x-y) + \sgn((1-x) - (1.5 - y)), \quad \mbox{and} \nonumber \\ 
v_2(x,y) \;=\; & -v_1(x,y).
\end{align}
To make the payoffs stay in the interval $[-1, +1]$, \citet{sion1957game} add
$+1$ to player $1$'s payoffs and $-1$ to player $2$'s payoffs so that 
utilities are given by:
\begin{equation}
u_1(x,y) = \begin{cases}
  -1 &  \mbox{if } x < y < x + \textstyle \half;
\\
  0  &  \mbox{if } x = y  \mbox{ or } y =  x + \textstyle \half \\ 
  1  &  \mbox{otherwise}.\\ \end{cases} 
\end{equation}
and $u_2(x,y) = -u_1(x,y)$.
Diagramatically, we can represent $u_1(\cdot,\cdot)$ as in Figure 1.  

\begin{figure}[h]
    \centering
    \usetikzlibrary{decorations.pathmorphing,decorations.text,arrows.meta}
\begin{tikzpicture}[scale=7]

  \draw[thick,fill=black,draw=black] (0,0) -- (1,0);
\draw[thick,fill=black,draw=black] (0,0) -- (0,1);
    \draw[thick,fill=black,draw=black] (1,0) -- (1,1);
        \draw[thick,fill=black,draw=black] (0,1) -- (1,1);

  \node[below left] at (0,0) {\footnotesize 0};
\node[below=20pt] at (0.5,0) {\footnotesize $x$};
  \node[left=20pt] at (0,0.5) {\footnotesize $y$};
  \node[below] at (1,0) {\footnotesize 1};
  \node[left] at (0,0.5) {\scriptsize $\frac{1}{2}$};
\node[above] at (0.5,1) {\scriptsize $\frac{1}{2}$};
  \node[left] at (0,1) {\footnotesize $1$};
  \draw[opacity=.3,fill=white] (0,0) -- (1,1) -- (1,0) -- cycle;
\draw[opacity=.3,fill=white] (0,0.5) -- (0,1) -- (0.5,1) -- cycle;
   \draw[opacity=.3,draw=white,fill=white] (0,0) -- (1,1) -- (0.5,1) -- (0,0.5)-- cycle;
\draw[line width=1pt,dotted,fill=black,draw=black] (0,0) -- (1,1);
  \draw[line width=1pt,dotted,fill=black,draw=black] (0,0.5) -- (0.5,1);
  \node at (0.15,0.89) {\scriptsize $u_1(x, y) \,= \,1$};
\node at (0.75,0.2) {\scriptsize$u_1(x, y) \,=\, 1$};
   \node at (0.2,0.45) {\scriptsize$u_1(x, y) \,=\, -1$};
\node at (0.6,0.85) {\scriptsize $u_1(x, y) = 0$};
\draw [thick,dotted,->,>={Latex[open]}] (0.6,0.80) to [out=-90,in=-45] (0.36,0.86);
\draw [thick,dotted,->,>={Latex[open]}] (0.6,0.90) to [out=85,in=135] (0.85,0.85);
\end{tikzpicture}
 \caption{The Sion-Wolfe Game \citeyearpar{sion1957game}}
    \label{fig:sionwolfe}
\end{figure}
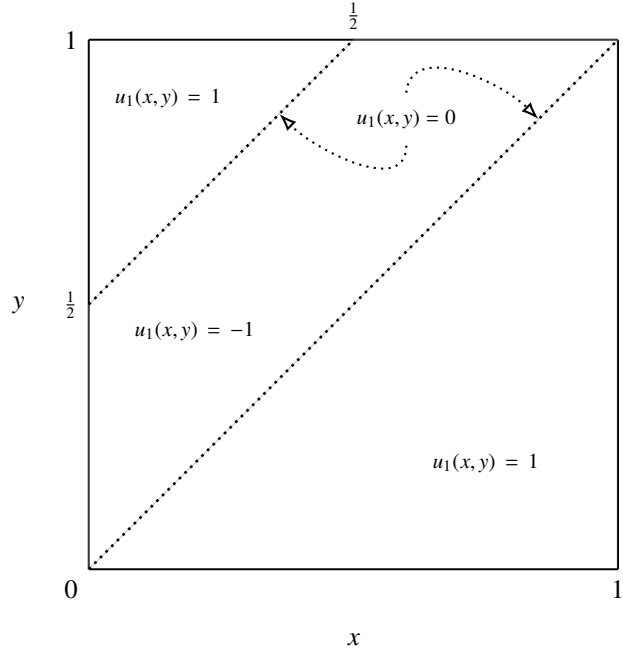 \medskip

\citet{sion1957game} show that this game has no countably additive
$\epsilon$-equilibria for a range of strictly positive $\epsilon$.
We give
two finitely additive equilibria that represent the limits of ``reasonable''
finitely supported equilibria for this game, providing a contrast with both
the Sion and Wolfe result and the literature on discontinuous games that has
assiduously avoided games of this sort.

The differences between the two finitely additive equilibria depend on details
of the hyperfinite action sets, equivalently on the details of the net of
finite approximations.
Player $2$'s advantage of $0.5$ in the second
battlefield means that fine details of the approximation around $(0.5,0.5)$
can matter.
Let $H_1$ and $H_2$ denote the two players' exhaustive
hyperfinite action sets, and throughout, let $h_i'$ denote player $i$'s
largest strategy strictly less than $0.5$ in $H_i$.

We analyze two cases, the fully symmetric one, and one of the asymmetric ones.
In the fully symmetric case, we suppose that $H_1 = H_2$, equivalently, that
the nets of finite approximations satisfy $F_1^\alpha = F_2^\alpha$.
In the
asymmetric case, we suppose that $h_1' < h_2'$, equivalently, that
$h_2^\alpha$, the largest elements of $F_2^\alpha$ below $0.5$, is larger than
the corresponding $h_1^\alpha$ in $F_1^\alpha$.
It is of particular note that we find equilibria with different values.
This
is consistent with the point of view that continuum games are not fully
specified until one has specified what large finite sets the continuum is
supposed to represent.

\subsubsection{The Symmetric Case}

Here is one kind of equilibrium for the Sion and Wolfe game.

\begin{lemma}\label{Sion-Wolfe lemma equal Hs} If the exhaustive hyperfinite
sets replacing $A_1$ and $A_2$ satisfy $H_1 = H_2$ and $h'$ is the largest
element of $H_i$ that is strictly less than $0.5$, then after iterated
deletion of weakly dominated strategies: player $1$ has three strategies, $0$,
$h'$, and $1$;
player $2$ has three strategies, $h'$, $0.5$, and $1$;  the
unique hyperfinite equilibrium is $(\eta^*_1,\eta^*_2) = \bigl( (\frac{1}{5},
\frac{1}{5}, \frac{3}{5}), (\frac{1}{5}, \frac{1}{5}, \frac{3}{5})\bigr)$;
and the
equilibrium utilities are $(+0.4,-0.4)$.  \end{lemma}

The $3 \times 3$ game just below has the unique equilibrium
given.\footnote{\citet{yanovskaya1970solution} and \citet{schervish1996fair}
give methods for choosing the payoff at $(h',h')$ according to utility
mid-point rules that can be unrelated to the limits of payoffs around the
point $(0.5, 0.5)$.
In this game, \citet{yanovskaya1970solution} finds an
equilibrium putting mass on the same three points but with probabilities
$(\eta^*_1,\eta^*_2) = ( (3\sqrt{2}-4, 3 - 2\sqrt{2}, 2 - \sqrt{2}), (3 -
2\sqrt{2}, 3\sqrt{2}-4, 2 - \sqrt{2}))$ and with equilibrium utilities
$(\sqrt{2}-1, 1 - \sqrt{2})$.
}

\medskip

\begin{table}[h!]
\centering
\begin{NiceTabular}[columns-width=15pt]{@{}p[c]{5pt}p[c]{5em}p[c]{5em}p[c]{5em}@{}}
  &  
$h'$       &  $\half$    &  $1$ \\[2pt] 
\specialrule{2pt}{3pt}{0.5pt}
\addlinespace[3mm] $0$ \Block[borders=right,line-width=1.5pt]{*-1}{}&  $(-1,+1)$  &  $(0,0)$    &  $(+1,-1)$  \\[8pt] 
$h'$ &  $( 0, 0)$  &  $(-1,+1)$  &  $(+1,-1)$  \\
\addlinespace[3mm] $1$ &  $(+1,-1)$  &  $(+1,-1)$  &  $( 0, 0)$  
\end{NiceTabular}
\end{table}
\medskip

The pair $(h',h')$ corresponds to the two players playing on the diagonal just
below and to the left 
of the point $(0.5, 0.5)$.  As described above in  
\S\ref{subsection just under}, we can also represent it using a net $\alpha
\mapsto (h_1^\alpha, h_2^\alpha)$ or with the product distribution on $[0, 1]
\times [0, 1]$ putting all of its mass on the diagonal just below and to the
left of the point $(0.5, 0.5)$.

\subsubsection{An Asymmetric Case}

Here we analyze the equilibria that arise when player $2$ can play closer to
but still below $0.5$ than player $1$ can.

\begin{lemma}[An asymmetric version]\label{Sion-Wolfe lemma}
If $H_1$ and $H_2$ are exhaustive hyperfinite sets for $A_1$ and $A_2$, and
$h_1' < h_2'$ are the largest elements in exhaustive hyperfinite action sets
less than $0.5$, then the sets of strategies that survive iterated deletion of
weakly dominated strategies is $\{0, 1\}$ For player $1$ and $\{h'_2, 1\}$ for
player $2$, and the finitely additive equilibrium corresponds to play of
$\eta_1^* = \frac{1}{3} \delta_{0} + \frac{2}{3} \delta_{1}$ for player $1$
and $\eta_2^* = \frac{1}{3} \delta_{h'_2} + \frac{2}{3} \delta_{1}$ and the
equilibrium utilities are $(\frac{1}{3}, - \frac{1}{3})$.
\end{lemma}

In this result, one can alternatively express $h_1' < h_2'$ as $h_1^\alpha <
h_2^\alpha$ where $h_1^\alpha$ and $h_2^\alpha$ are the largest elements less
than $0.5$ in an exhaustive net $\alpha \mapsto (F_1^\alpha, F_2^\alpha)$ in
$A_1 \times A_2$.

\subsubsection{Different Equilibrium Payoffs in a Zero-Sum Game}

If it is not an article of faith among game theorists that zero-sum should have
just one value, then it is at least a widely held preconception.
The present
analysis suggests that this is mistaken.  If one is committed, as we are here,
to finite approximations, then one might suppose that the modeler's decision
that both player's strategies can be well represented by the usual continuum
means that any finite approximations must be equal to each other.
If so, one
must choose the first equilibrium.  If this supposition is not correct, then
the continuum assumption has not provided enough detail.
But equality or
inequality of the approximations is not a choice for which theory provides any
simple guidance.
\citet[Example 2.5]{stinchcombe2005nash} gives a game where
preserving the continuum strategic structures in finite approximations cannot
be done unless we have equality of the finite approximations and we have
inequality of the finite approximations.

\subsection{Just Under/Just Over}\label{subsection just under}

The Sion and Wolfe game makes plain why the representation of ``just under''
matters: the equilibrium depends not merely on whether a player concentrates
mass below $r$, but on which player can approach $r$ more closely.
Three
representations of this idea have appeared in the analysis --- exhaustive nets,
exhaustive hyperfinite sets, and finitely additive limits.
This subsection
shows that all three are equivalent. \textit{Mutatis mutandis}, the same analysis
covers numbers just over an $r \in [0,1)$.

\subsubsection{Exhaustive Nets and Finitely Additive Limits}

Suppose that $\alpha \mapsto F_1^\alpha$ and $\alpha \mapsto F_2^\alpha$ are
exhaustive nets of subsets of $[0,1]$, and pick $\beta$ so that $r \in
F_i^\alpha$ for all $\alpha \succsim \beta$.
For $\alpha \succsim \beta$, let
$h_1^\alpha$ and $h_2^\alpha$ be the largest elements of $F_1^\alpha$ and
$F_2^\alpha$ strictly less than $r$.
Three cases arise: $h_1^\alpha <
h_2^\alpha < r$, in which player $2$ can approach $r$ from below more closely
than player $1$;
$h_1^\alpha = h_2^\alpha < r$, in which the players can
approach equally closely;
and $h_2^\alpha < h_1^\alpha < r$, in which player
$1$ can approach more closely.

If the equilibrium for the game played on
$F_1^\alpha \times F_2^\alpha$ assigns probability $\mu_1^\alpha$ to
$h_1^\alpha$ and $\mu_2^\alpha$ to $h_2^\alpha$, then $\mu^\alpha :=
\mu_1^\alpha \times \mu_2^\alpha$ assigns probability $\mu_1^\alpha \cdot
\mu_2^\alpha$ to the pair $(h_1^\alpha, h_2^\alpha)$.
By compactness of the
set of finitely additive probabilities, accumulation points of the net $\alpha
\mapsto \mu^\alpha$ exist.

The finitely additive limits retain information that the countably additive
limits discard.
If $\mu^\alpha(\{h_1^\alpha, h_2^\alpha\})$ is bounded away
from $0$, then: in the first case, $\mu := \lim_\alpha \mu^\alpha$ puts
positive mass on each of the sets $\{(x,y): r-\epsilon < x < y < r\}$;
in
the second, on each of the sets $\{(x,y): r-\epsilon < x = y < r\}$;
and in
the third, on each of the sets $\{(x,y): r-\epsilon < y < x < r\}$.
The
countably additive limit, by contrast, must in all three cases place the limit
mass on $(r,r)$ --- losing precisely the information about which player was
closer to $r$.

Since $\mu = \lim_\alpha \mu^\alpha$ is a finitely additive limit, the
integrals of any bounded function converge by definition.
The equilibrium
utility is $\lim_\alpha \int u(a)\,d\mu^\alpha(a)$, and the utility to any
deviation $b_i \in A_i$ is $\lim_\alpha \int u(a \backslash
b_i)\,d\mu^\alpha(a)$.
When utilities depend on which player is closer to $r$
from below, the finitely additive limit reflects this distinction;
the
countably additive limit does not.

\subsubsection{Exhaustive Hyperfinite Sets and Standard Parts}

Suppose that $H_1$ and $H_2$ are exhaustive hyperfinite subsets of $[0,1]$;
the ultrapower construction is given in the appendix. Exhaustiveness guarantees
$r \in H_1$ and $r \in H_2$ for any $r \in [0,1]$.
Let $h_1'$ and $h_2'$
denote the largest elements of $H_1$ and $H_2$ strictly less than $r$.
The
same three cases arise: $h_1' < h_2' < r$, in which player $2$ can approach
$r$ from below more closely;
$h_1' = h_2' < r$, in which the players can
approach equally closely;
and $h_2' < h_1' < r$, in which player $1$ can
approach more closely.
What the hyperfinite representation adds is precision
about the distances: $(r - h_i') > 0$ and $(r - h_i') \simeq 0$ --- the
distances are strictly positive but infinitesimal.

If the equilibrium for the
game played on $H_1 \times H_2$ assigns probability $\eta_1$ to $h_1'$ and
$\eta_2$ to $h_2'$, then $\eta := \eta_1 \times \eta_2$ assigns probability
$\eta_1 \cdot \eta_2$ to the pair $(h_1', h_2')$.
The standard part of $\eta$
is a finitely additive probability on $[0,1] \times [0,1]$.

The standard parts retain information that the countably additive limits
discard.
In the first case, the standard part $\mu$ puts mass infinitesimally
close to $\eta_1(h_1') \cdot \eta_2(h_2')$ on each of the sets $\{(x,y):
r-\epsilon < x < y < r\}$;
in the second, on each of the sets $\{(x,y):
r-\epsilon < x = y < r\}$;
and in the third, on each of the sets $\{(x,y):
r-\epsilon < y < x < r\}$.
The countably additive limit, by contrast, must in
all three cases place the limit mass on $(r,r)$ --- losing precisely the
information about which player was closer to $r$.

Since $\mu$ is the standard part of $\eta$, the integrals of any bounded
function are infinitesimally close by definition.
The equilibrium utility
$\int_A u(a)\,d\mu(a)$ is the unique number infinitesimally close to $\int_H
u(h)\,d\eta(h)$, and the utility to any deviation $b_i \in A_i$ is the unique
number infinitesimally close to $\int_H u(h \backslash b_i)\,d\eta(h)$.
When
utilities depend on which player is closer to $r$ from below, the standard
part reflects this distinction;
the countably additive limit does not.

The two representations are thus two languages for the same finitely additive
limit --- exhaustive nets and exhaustive hyperfinite sets arrive at identical
conclusions about which player is closer to $r$, how the mass is distributed,
and what the equilibrium utilities are.

\subsection{Problems with Iterated Integrals}
The finitely additive framework resolves the iterated integral problem that
has bedeviled the prior literature --- but it does not dissolve it.
For a
small class of utility functions $u: X \times Y \to [-B,+B]$, the iterated
integrals are equal:
\begin{equation}\label{fubini property}
\int_X \left[ \int_Y u(x,y)\,dq(y) \right] dp(x) =
\int_Y \left[ \int_X u(x,y)\,dp(x) \right] dq(y),
\end{equation}
for all finitely additive total probabilities $p$ and $q$.
Measure-theoretic,
algebraic, functional analytic, and finite approximability characterizations
of this class can be found in \citet{harris2005nearly}.
For the general case,
the product measure $p \times q$ is well-defined only on finite unions of
rectangles $A \times B \subset X \times Y$;
by Lemma~\ref{lemma about
extensions}, the set of extensions to all subsets of $X \times Y$ is compact
and convex, and the set of integrals with respect to those extensions is
compact and convex as well.
The set-valued integrals that arise in this
general case are studied systematically in \citet[\S4]{stinchcombe2005nash}.

Two consequences of this structure are examined here. The first ---
Wald's largest integer game --- shows that the set of finitely additive
equilibria can be strictly larger than the set of limits of approximate
equilibria along exhaustive nets: the finitely additive equilibrium
correspondence is not, in general, a representation of the limits of finite
approximations.
The second --- a Karlin game --- redeems the promise of
\S\ref{subsect fa strats}: it shows precisely how
both players optimizing their iterated integrals can deliver utilities that do
not sum to zero, exposing the source of Karlin's error.

\subsubsection{Wald's Largest Integer Game}\label{subsection largest integer}

At the end of \citet[\S 1]{wald1945generalization}, we find the two person
with $A_1 = A_2 = \integers$ and zero-sum payoffs $u(n_1,n_2) = (\sgn(n_1-n_2),
\sgn(n_2-n_1))$.
This is the ``pick the largest integer'' game, and it
clearly has no value in countably additive strategies.
The fact that the
integral is not well-defined for products of finitely additive mixed
strategies has led the literature to conclude that the game has no
equilibrium.

\begin{lemma}\label{lemma about the Wald game}
If $\mu_1$ and $\mu_2$ are purely finitely additive probabilities on $A_1$ and
$A_2$ respectively, then: 
\begin{enumerate}[\hspace*{1mm}\normalfont (a)]

  \item  every element of the compact set of product extensions of
$(\mu_1,\mu_2)$ is an equilibrium, the associated set of equilibrium
payoffs is the convex set $\{(+r, -r): r \in [-1, +1]\}$;
and 

  \item  the set of payoffs to the finitely approximable equilibria is 
$\{ (+1,-1),$ $(0, 0),$  $(-1,+1)\}$.
\end{enumerate}
\end{lemma}

Point (b) makes it a bit clearer how to think of at least some of the
equilibria for the game --- either the players are equally good at naming
large integers, which leads to the tie payoffs, $(0, 0)$, or one of them is
better than the other, which lead to the other two payoff vectors.

\subsubsection{On Karlin Games}\label{subsection Karlin's claims}

For ease of comparison with \citet[Theorem 10, p.\ 152-3]{karlin1950operator}
(and its generalization in \citet[Theorem 3]{karlin1953theory}), we give this
as a two person zero-sum game played on $X \times Y = [0, \infty] \times [0,
\infty]$ with strategies vectors $(x, y) \in X \times Y$ for players $1$ and
$2$.
Identifying $[0, \infty]$ with $[0, 1]$ via a strictly increasing smooth
homeomorphism returns us to Karlin's setting of games on products of the unit
interval.
To specify the utility functions in the example, we use two
functions: (1) $\Phi(r) = 2(F(r) - \half)$ where $F(r) = e^r/(1+e^r)$ is the
logistic cdf;
and (2) $\psi(x) = 1 - e^{-r}$, the exponential cdf.

For $x, y \in [0, \infty)$, the utilities are
\begin{equation}\label{first part Karlin utilities}
u(x,y) = 
  \Bigl(\Phi(x-y) + [\psi(x) - \psi(y)],\; \Phi(y-x) + [\psi(y) - \psi(x)]\Bigr) ,
\end{equation}
and for the other boundaries of the product of the action sets, utilities are
\begin{equation}\label{second part Karlin utilities}
u(x,y) = \begin{cases}
  (-(2 + \Phi(y)), +(2 + \Phi(y))) 
              & \ \mbox{if} \ x = \infty, y < \infty, \\
  (+(2 +\Phi(x)), -(2 + \Phi(x)) 
              & \ \mbox{if} \ 
x < \infty, y = \infty,  \\
  (0, 0) 
              & \ \mbox{if} \ x = \infty, y = \infty.\\
\end{cases}
\end{equation}
For $x, y \in [0, \infty)$, the $\Phi(x-y)$ and $\Phi(y-x)$ parts of the
utility functions in (\ref{first part Karlin utilities}) give a ``pick the
largest number'' game on $[0, \infty) \times [0, \infty)$, while  the
$[\psi(x) - \psi(y)]$, $[\psi(y) - \psi(x)]$ parts of the utility function
ensure that larger numbers are strictly better no matter what the finitely
additive probability that other player uses.
The parts of the utility
functions given in (\ref{second part Karlin utilities}) guarantee that 
$x = \infty$ or $y = \infty$ are strictly dominated strategies.

If we evaluate the players' utilities using the iterated integrals, there is a
vector $(p^*,q^*)$ of mutual best responses.
They deliver the infeasible
utility vector $(-1, -1)$ because the iterated integrals do not represent the
payoffs.

\begin{lemma}\label{lemma on Karlin games}
The set of limit equilibrium payoffs for finite approximations to this game 
is $\{(r,-r): r \in [-1, +1]\}$, but if $p$ is a probability on $X$, then any
$q^*$ that solves 
\begin{equation}
 \max_{q}  \int
\left[  \int_X u_2(x,y)\,dp(x) \right] \,dq(y)
\end{equation}
satisfies $q((n,\infty)) \equiv 1$, and if  $q^*((n,\infty)) \equiv 1$, then
any $p^*$ that solves
\begin{equation}
 \max_{p}  \int \left[  \int_Y u_1(x,y)\,dq^*(y)
\right] \,dp(x)
\end{equation}
delivers a utility $-1$.
\end{lemma}

Except for utility functions that return the analysis to the
\citet{glicksberg1952further} setting of compact and continuous games, one
cannot use iterated integrals to evaluate the payoffs of finitely additive
strategies.

\section{Countably Additive Utility Equivalences}\label{sec:specialdiscontinuities}

In this section, we investigate, for finite player games with compact metric
spaces of actions, the class of Borel measurable utility functions with
discontinuities special enough that the equilibrium utilities are the same for
finitely additive total equilibria and countably additive Borel equilibria.
At a conceptual level, the difficulties arise from a peculiar mismatch:  most
of the literature on infinite games has used, as we do extensively in this
section, the weak$^*$ topology for countably additive probabilities.
But the
utility functions do not integrate these probabilities continuously when we
use this topology with the following \textit{caveat}: they do converge if the
limit probability puts mass zero on the discontinuities of the utility
function.\footnote{See, for example,\ \citet[Theorem 5.1, p.\
30]{billingsley1968convergence} for this.
To our knowledge,
\citet{yanovskaya1964minimax} was the first to use this logic for zero-sum
games and \citet[Theorem 4]{dasgupta1986Itheory} was the first for more
general payoff functions.}

\subsection{Continuous Equivalence}

The following is the central concept used in this section.
Note that it
allows for comparisons of the class of probabilities defined for all subsets
that we use here and the class of Borel probabilities.
\begin{dfnn} 
Probabilities $p$ and $p'$ on a metric space are \textbf{continuously
equivalent} if they integrate all bounded continuous functions to the same
number.
\end{dfnn}

It is a classic result that two countably additive Borel probabilities on a
metric space are continuously equivalent if and only if they are equal.
For
any total finitely additive probability on a compact metric space, the Riesz
representation theorem guarantees the existence of a unique countably additive
probability that is continuously equivalent.
The differences can be seen in
the following.

\begin{exx}
On the compact metric space $[0, 1]$, a finitely additive $p$ is continuously
equivalent to the countably additive point mass on $0$, $q = \delta_{0}$ if
and only if it puts mass $1$ on every half-open set $[0, \epsilon)$.
The
purely finitely additive probabilities in this class are the ones that put
unit mass on every open set $(0, \epsilon)$.
Such $p$ are distinguished from
$q$ by the upper semi-continuous function $f(\cdot)$ with $f(0) = 1$ and $f(x)
= 0$ for $x > 0$ because $\int f(x)\,dq(x) = 1 > 0 = \int f(x)\,dp(x)$.
\end{exx}

Thinking of $q$ in this Example as the equilibrium in a $1$-person game, we
see that if an equilibrium $q$ puts mass on the discontinuities of the utility
function, then the continuously equivalent finitely additive $p$'s need not be
equilibria.
Continuous equivalence can also miss phenomena of game-theoretic
importance.
\begin{exx}\label{ignore Pareto dominant eqa}
Player $1$ picks an $x \in [0, 1]$ and an action $a$ in the two point set
$\{\alpha, \beta\}$.
Player $2$ picks a $y \in [0, 1]$.  The utilities for
both players are given by the same function,
\begin{align}
u((x,\alpha),y) & \quad=\quad (2-x)(2-y)\quad \qquad \mbox{and} \ \\[5pt]
u((x,\beta),y) & \quad=\quad   \nonumber
     \begin{cases}
    1.5 (x+1)(y+1)   & \mbox{if } (x,y) \ll (1,1) \\
    (0, 0)       & \mbox{otherwise} .
\end{cases} 
\end{align}
\end{exx}\smallskip

\noindent \textbf{Analysis}.  Play of $((0,\alpha),0)$ with probability $1$ is
the unique countably additive equilibrium that yields equilibrium utilities
$(4,4)$.
If $p^*$ is a vector of purely finitely additive probabilities in
which both players play actions ``just below'' $1$ and player $1$ combines
this with the action $\beta$, we have a finitely additive equilibrium that
yields equilibrium utilities $(6,6)$.
The continuously equivalent countably
additive $q$ is point mass on $(1, 1)$, and this yields the non-equilibrium,
minimal possible utilities, $(0,0)$.\footnote{Replacing the strategy sets $[0,
1]$ by the half-open interval $[0, 1)$ gives a non-compact game satisfying
better reply security.
One can understand the finitely additive equilibrium
as having compactified both strategy sets with a pair of points just below $1$
but above $1-\epsilon$ for all $\epsilon > 0$ and yielding utilities $(6,6)$.}

\subsection{Continuous Equivalence for Equilibria}

We study the behavior of all of the finitely additive probabilities that are
continuously equivalent to a given countably additive Borel probability that
is an equilibrium.
We then study the behavior of the countably additive Borel
probability that is continuously equivalent to a finitely additive
equilibrium.
In both cases, we are interested in conditions guaranteeing that
the payoffs are the same.

The arguments pass through two Lemmas of independent interest.  For
perspective on the central role that upper and lower semi-continuous functions
will play, both in the first Lemma and in the two results, note that for
countably additive $q_n$ and $q$, $\int g \,dq^n \rar \int g\,dq$ for all
bounded continuous functions $g$ if and only if for all bounded upper
semi-continuous functions, $\limsup_n \int f\,dq^n \leq \limsup_n \int f\,dq$,
with the reverse inequality for lower semi-continuous functions.

\subsubsection{Two Lemmas}

We will use the following to compare the utilities of deviations against
finitely additive strategies and the countably additive continuously
equivalent probability.

\begin{lemma}\label{lemma for deviations}
For $X$ a compact metric space and $f:X \rar \reals$ a bounded upper
semi-continuous function, if $p$ is a total finitely additive probability on
$X$ and $q = ca(p)$ is its countably additive version, then $\int f(x)\,dq(x)
\geq \int f(x)\,dp(x)$, and the inequality reverses if $f$ is lower rather
than upper semi-continuous.
\end{lemma}

We will use the following to analyze the set of games for which the finitely
additive equilibria and their countably additive versions deliver the same
utilities to the agents.

\begin{lemma}\label{lemma for the same utility}
For $X$ a compact metric space, $f:X \rar \reals$ a bounded Borel measurable
function, and $q$ a countably additive Borel probability, if $q$ puts zero
mass on the closure of the discontinuities of $f$, then $\int f\,dq = \int
f\,dp$ for all finitely additive $p$ that are continuously equivalent to $q$.
\end{lemma}

\subsubsection{The Utility Equivalences}

We now treat the finitely additive probabilities continuously equivalent to a
countably additive equilibrium.

\begin{thmm}\label{thm:ctsequivofacaeqm}
Suppose that $q^*$ is a countably additive equilibrium for a finite player
game with compact metric spaces of actions.
If $q^*$ puts zero mass on the
closure of the discontinuities of $u:A \rar \reals^I$ and for all $\ii$ and
all $b_i \in A_i$, the mapping $a \mapsto u_i(a \backslash b_i)$ is upper
semi-continuous, then every finitely additive $p$ that is continuously
equivalent to $q^*$ is an equilibrium that gives the same expected utility
payoffs as $q^*$.
\end{thmm}

We now treat the countably additive continuous equivalent of a finitely
additive equilibrium.

\begin{thmm}\label{thm:ctsequivofafaeqm}
Suppose that $p^*$ is a finitely additive equilibrium for a finite player game
with compact metric spaces of actions and that $q^*$ is the countably additive
probability continuously equivalent to $p^*$.
If $q^*$ puts mass $0$ on the
closure of the discontinuities of $u:A \rar \reals^I$ and for all $\ii$ and
all $b_i \in A_i$, the mapping $a \mapsto u_i(a \backslash b_i)$ is lower
semi-continuous, then $q^*$ is an equilibrium that gives the same expected
utility payoffs as $p^*$.
\end{thmm}

\section{Summary and Future Directions}\label{sec:summary}

There has always been a large gap between the theory of finite games and the
theory of infinite games.
Finite games always have equilibria and the set of
equilibria depends upper hemicontinuously on the specification of the game.
But if one chooses to use the usual textbook models of the continuum and of
countably additive Borel probabilities, neither of these foundational results
hold.\footnote{ Indeed, if one believes that these textbook models are a
necessary part of the analysis of infinite games, then one might be convinced
by the \citet[\S3]{jackson2009non} contention that the equilibria of large
finite approximate models are misleading if models using the usual continuum
do not have equilibria.
} Our contention is that the gaps between the finite
models and the infinite models are an artifact of this choice because these
objects do not retain information about how one arrived at the limit.

We have shown that the gap disappears if strategies are understood as finitely
additive mixtures as we use them.
That being said, the literature has many
objections to, and many arguments for, the use finitely additive
probabilities.
Both need to be understood as being specific to the context.
If the context dependent objections to finitely additive mixtures are to have
bite in game theory, they must provide a principled repudiation of the
validity of finite approximations to games, and they must also be strong
enough to overcome the strengths of the theory.

The following strengths are the core of our contention that finite additivity is the correct choice of model for mixed strategies in the study of infinite games: equilibrium existence, equilibrium stability, and equilibrium operationalizability are delivered by a single framework.
Beyond this, the
examples in \cref{sec:gamesontheline} show that this last set of
equilibria can be very easy to analyze in several of the classical games that
are thought of as difficult or impossible.\footnote{Though we do recognize
that ``ease,'' like ``beauty,'' is in the eye of the beholder.} We also
believe that future work on continuum extensive form games and infinite player
games will strengthen these arguments.

\subsection{Continuum Extensive Form Games}

There are several well-studied examples of extensive form games with infinite
sets of actions where early choices upper- but not lower-hemi\-con\-tinuously
determine the set of later equilibrium payoffs in such a fashion as to
preclude the existence of countably additive equilibria.
Perhaps the best
known example is the analysis of competing principals in
\citet{myerson1982optimal}.
In such games, the ability to model players
choosing ``just under/over'' the point at which the payoffs jumps down
provides a direct and immediate solution to the non-existence problems.
Indirect solutions at indifference points for later players that capture this
sort of phenomenon can also involve the cheap talk additions to extensive form
games as in \citet{manelli1996cheap} or \citet{jackson2002communication}.

Another issue for countably additive mixtures in extensive form games is the
``disappearance'' of information in the limit.
\citet[p.\
497]{myerson2020perfect} write that 
\begin{quote}
$\ldots$ the difficulty is that the randomized signals upon which players
coordinate their actions along the sequence can, in the limit, have
distributions that degenerate to a point, leaving the players without access to
the necessary coordination device.
\end{quote}
This kind of disappearance of information at a point is an artifact of the
insistence on countable additivity and the routine application of the weak$^*$
topology for countably additive Borel probabilities.
That topology was
developed for continuous problems and these are discontinuous problems.  The
\citet[\S6.4]{myerson2020perfect} solution is to use finitely additive beliefs.

Perhaps the clearest expressions of the preservation of information in the
limit are the finitely approximable equilibria in the \citet{sion1957game}
game.
There the limit objects retain the information about the relative size
of the players choices along the approximating nets.
The preservation of
information in the limit is far more general than this indicates,
\citet[Corollary 5.1]{stinchcombe2023direct} shows that finitely additive
probabilities that put unit mass one on \textit{every} interval $(0,
\epsilon)$ can encode any distribution on the usual model of the continuum.

There remains, however, a large conceptual difficulty to be overcome for
infinite extensive form games.
Consider a game model in which players sequentially choose actions in, say,
$[0, 1]$, and the later actions are chosen on the basis of signals that are
continuous functions of the early play.
There are two very different options
for finite approximations to this game.
One could, (i), exhaustively replace
each space $[0, 1]$ and analyze the resulting net of finite extensive form
games, or (ii), one could start with the set of pure strategies as measurable
functions from signals to later choices and replace the set of measurable
functions with an exhaustive net of finite sets.
The choice matters.
\citet[Example 2.5, p.\ 340]{stinchcombe2005nash} gives a game for which the
two modeling strategies give substantively different strategic structures.

If one adopts strategy (ii) to analyze continuum extensive form games, there is
another advantage to our approach that is perhaps not apparent.
It arises
from the fact that we allow for all bounded utility functions on $A = \ppi
A_i$, not just the ones that are measurable with respect to some product
$\sigma$-field.
This may seem like generality for the sake of generality, but
it plays a crucial role in the analysis of extensive form games with continuum
choices and signals.

In the simplest class of these extensive form games, player $1$ picks an
action $a_1$ in the unit interval, $A_1 = [0, 1]$, player $2$ observes the
choice and responds with a choice of his/her own in the unit interval.
Here
player $2$'s set of strategies is the set of functions from the unit interval
to itself, $f_2 \in F_2 = [0, 1]^{[0, 1]}$.
The outcome associated with the
choices $(a_1,f_2)$ is $(a_1,f_2(a_1))$.  To handle countably additive
randomization, one uses the Borel $\sigma$-field on $A_1$ and restricts $F_2$
to the be set of Borel measurable functions.
At this point, to make the model
work with countably additive random strategies, one needs to find a
$\sigma$-field on $F_2$ so that the outcome mapping, $(a_1,f_2) \mapsto
(a_1,f_2(a_1))$, is jointly measurable with respect to the product
$\sigma$-field on $A_1 \times F_2$.
\citet{aumann1961borel} shows that this
can be done, but only if one further restricts $F_2$ to be a bounded Banach
class of measurable functions.

By contrast, the mapping $(a_1,f_2) \mapsto (a_1,f_2(a_1))$ is measurable in
our context because the measurable structure on $A_1 \times F_2$ is, by
assumption, the class of all subsets.
Questions about the existence of a
product $\sigma$-field with special properties are simply irrelevant.
And for
continuum extensive form games where strategies are functions from one
continuum to another, no product $\sigma$-field is ever going to be sufficient
--- not even the event that both players use the same strategy is measurable.

\begin{lemma}\label{lemma product sigma fields don't cut it}
There is no $\sigma$-field $\mcalX$ on $X = \mathbb{R}^{\mathbb{R}}$ for which
the diagonal, $\{(x,x): x \in X\}$ belongs to the product $\sigma$-field
$\mcalX \otimes \mcalX$.
\end{lemma}

\subsection{Infinite Population Games}

Though we have not made much use of it, Theorem \ref{thm:equilibriaexist} also shows that equilibria exist for games that use \textit{any}
infinite set to model the set of players.
We have several preliminary results
but we do not have the full picture of how that result relates to the
equilibria of the continuum population game models first studied by
\citet{schmeidler1973equilibrium}.
A summary of the extensive literature
that followed on this can be found in \citet{khan2002non}.
The analyses have
(a) worked with a variety of countably additive nonatomic probability
structures for the space of players, but crucially, (b), until very recently,
that work has posited that the agents can and do correctly observe the true
population distribution of actions.\footnote{Recent advances include
\cite{cerreia2022equilibria}, who model players as correctly observing the
true distribution of the actions taken by those in the player's peer or
comparison group, and \citet{frick2022dispersed}, who model players as using
the biased sample of the people that they interact with is representative of
the entire population.}

As to (a), extending the countably additive probability on a limited
$\sigma$-field of sets of players to a larger one 
is always possible.  A
central question is how sensitive the set of equilibria is to the choice of
extension.
As background for such an investigation, we have used the
\citet{mas-colell1984theorem} distributional equilibrium approach to continuum
population games.
In that approach, one defines an equilibrium as a joint
distribution of agent characteristics and actions having the mutual best
response property.
Because one focuses on the induced joint distribution on
player characteristic-action pairs, the measure theoretic differences in the
population model play a much smaller role.
In particular, this smaller role
means that the entire issue of whether or not an equilibrium in pure
strategies exists is thoroughly submerged.\footnote{For continuum population games with an atomless countably additive population measure, \citet[Theorem
2]{mas-colell1984theorem} gives a pure strategy equilibrium existence result when the identical  action set is  finite, \citet{khan1995marriage} extend this result to a countably infinite set  of actions, and  \citet{rath1995nonexistence} give counter-examples when those conditions are violated.
\citet{he2017modeling} show that when one restricts the countably additive population measure to have additional saturation properties, more general pure strategy existence results hold.}

Every countably additive Borel probability $q$ on a Polish space is
\textbf{tight}, that is, for each $\epsilon > 0$, there is a compact
$K_\epsilon$ such that $q(K_\epsilon) > 1 - \epsilon$.
A (possibly total)
finitely additive probability $p$ on a Polish space is \textbf{near tight} if
for all $\epsilon > 0$, there is a compact $K_\epsilon$ such that for every
$\delta$-ball $K_\epsilon^\delta \supset K_\epsilon$, $p(K_\epsilon^\delta) >
1 - \epsilon$.
Provided that the distribution of agent characteristics in an
infinite player game is near-tight, one can show that exact countably additive
distributional equilibria exist, and that a finitely additive joint
distribution on characteristics and actions is an equilibrium if and only if
it is continuously equivalent to the countably additive equilibrium.
Going
back to the issue of the many finitely additive extensions, we strongly
conjecture that the finitely additive extensions of the continuum population
model give rise to the different continuously equivalent equilibria as exact
finitely additive equilibria.

As to (b), the assumed correctness of how the population sees and interprets
what is happening in the world seems far too limiting for present day
phenomena, even when people use non-representative samples as in
\cite{cerreia2022equilibria} and \citet{frick2022dispersed}.
We have begun the
study of models in which evidence and data is selectively curated for each
individual by the advanced pattern recognition software currently deployed by
profit maximizing social media firms that value addiction of their customers
over the accuracy of what is presented.
We start the analysis knowing that
equilibria exist, and that they can be both interpreted and analyzed as limits
of equilibria for finite approximations to the strategic situation being
modeled.

\bibliographystyle{elsarticle-harv}
\bibliography{LargeGames}

\pagebreak

\appendix

\section{Compactification and Finitely Additive Probabilities}

\subsection{Compactifications}

Fix a non-empty set $X$ and a set $\mathbb{F}$ of bounded functions $f:X \rar
\reals$.
Assume that the set $\mathbb{F}$ separates points in $X$, that is,
for $x \neq y \in X$, there is some $f \in \mathbb{F}$ such that $f(x) \neq
f(y)$.\footnote{Alternatively pass to the space of $\mathbb{F}$-equivalance
classes in $X$: define $x \sim_{\mathbb{F}} y$ if $f(x) = f(y)$ for all $f \in
\mathbb{F}$;
and replace $X$ with the set of $\mathbb{F}$ equivalence
classes.} Define the $\mathbb{F}$-topology by the convergence of a net of
points $x^\alpha$ in $X$ by $x^\alpha \rar_{\mathbb{F}} x$ if $f(x^\alpha)
\rar f(x)$ for all $f \in \mathbb{F}$.
The
\textbf{$\mathbb{F}$-compactification} of $X$ is a compact Hausdorff space
$\widehat{X}_\mathbb{F}$ in which there is a homeomorphic imbedding $\varphi:X
\rar \widehat{X}_\mathbb{F}$ with $\varphi(X)$ a dense subset of
$\widehat{X}_\mathbb{F}$.
There are many constructions available, the
following is canonical.

For each $f \in \mathbb{F}$ let $I_f$ be the interval $[\inf_{x \in X},
\sup_{x \in X} f(x)] \subset \reals$ and let $Y$ be the product space
$\times_{f \in \mathbb{F}} I_f$.
Associate with each $x \in X$ the vector
$\varphi(x) := (f(x))_{f \in \mathbb{F}}$ in $Y$.
The definition of the
product topology on $Y$ guarantees that $x^\alpha \rar_{\mathbb{F}} x$ if and
only if $\varphi(x^\alpha) \rar \varphi(x)$ in the space $Y$.
Tychonov's
theorem guarantees that $Y$ is compact in the product topology.
The closure
of set of vectors, $\varphi(X) \subset Y$, defines compact Hausdorff space
$\widehat{X}_\mathbb{F}$, known as the $\mathbb{F}$-compactification of $X$.

We identify the set $X$ with its homeomorphic image as a subset of
$\widehat{X}_\mathbb{F}$.
With this identification, a defining property of
$\widehat{X}_\mathbb{F}$ is that every $f \in \mathbb{F}$ has a unique
continuous extension from $X$ to $\widehat{X}_\mathbb{F}$.
Specifically, the
function $x \mapsto f(x)$ from $X \subset \widehat{X}_\mathbb{F}$ to $[-B,
+B]$ is continuous, and, being defined on a dense subset of a compact
Hausdorff space, it has a unique continuous extension, $\widehat{f}:\varphi(X)
\rar \reals$.

A crucial property is that for each $f \in \mathbb{F}$, the problem $(*) \
\max_{x \in X}\, f(x)$ has a solution in $\widehat{X}_{\mathbb{F}}$:  let
$x^\alpha$ be a net of points in $X$ with $f(x^\alpha) \uparrow \sup_{x \in X}
f(x)$;
any net $x^\alpha$ of points in $X$  with $f(x^\alpha) \uparrow
\sup_{x \in X} f(x)$ has a closed non-empty set of accumulation points;
any
element of this set represents the limits of approximate solutions to $(*)$.

\subsection{Finitely Additive Probabilities}

This construction leads to the \citet{yosida1952finitely} representation of
finitely additive total probabilities on a set $X$.
Let $\mathbb{F}$ denote
the set of all indicator functions $1_B(\cdot)$ for $B \subset X$.
The
continuous extension of $f = 1_B \in \mathbb{F}$ takes on only the values $0$
and $1$.
It is denoted $\widehat{f} = 1_{\widehat{B}}$ where $\widehat{B}$ is
the set of $\widehat{x} \in \widehat{X}_{\mathbb{F}}$ for which
$\widehat{f}(\widehat{x}) = 1$.

For each $x \in X$, the set $\{x\}$ is $\mathbb{F}$-open and the extension of
$1_B$ with $B = \{x\}$ is continuous, and takes the value $1$ only at the
point $x \in \widehat{X}_\mathbb{F}$.
Being a union of open sets, $X$ is an
open subset of $\widehat{X}$, and the \textbf{halo} of $X$ is the closed,
hence compact, space $\widehat{X}_{\mathbb{F}} \setminus X$.
The countably additive point masses on the class of all subset of $X$ are the
probabilities $\delta_x$ satisfying $\delta_x(B) = 1_B(x)$ for all $B \subset
X$.
The purely finitely additive probabilities on the class of all subsets
of $X$ can be identified with the countably additive point mass on the halo.
By the Riesz representation theorem for compact Hausdorff spaces, each total
probability $p$ on $X$ has a unique countably additive extension,
$\widehat{p}$, to $\widehat{X}$, with $\widehat{p}(\widehat{X} \setminus
\varphi(X))$ being the weight of the purely finitely additive part of $p$.

\section{Exhaustive Hyperfinite Sets}

We give two short developments of nonstandard objects, including hyperfinite
sets.
The first one is based on elementary real analysis intuitions grounded
in sequences.
The second one is grounded in nets, and this is needed for the
exhaustiveness property.

\subsection{The Sequence Ultrapower Construction}

The starting point is a purely finitely additive $\{0, 1\}$-value probability,
denoted $\mu$, on $\mcalP(\integers)$, the class of all subsets of the
integers.
Letting $S_N$ denote the closure of the set of points masses
$\{\delta_n: n \geq N\}$ in the set of finitely additive probabilities, any
element of $\cap_N S_N$ is a purely finitely additive zero-one probability.

\begin{dfnn}
A function $\mu:\mcalP(\integers) \rar [0, 1]$ is \textbf{purely finitely
additive, zero-one probability} if 
\begin{enumerate}[\hspace*{5mm} \normalfont (1),leftmargin=1em]
  \item  for all $A \subset \integers$, $\mu(A) = 0$ or $\mu(A) = 1$;
  \item  $\mu(A \cup B) = \mu(A) + \mu(B)$ for all disjoint $A,B \subset
\integers$;
  \item  $\mu(\integers) = 1$;
and 
  \item  $\mu(A) = 0$ if $A \subset \integers$ is finite.
\end{enumerate} 
\end{dfnn}

For any set $X$, $X^\integers$ denotes the set of sequences, understood as
mappings, $n \mapsto x_n$, from $\integers$ to $X$.
We pick a purely finitely
additive, zero-one probability, $\mu$, and for our purposes, it will not
matter which one.
The properties of $\mu$ just given guarantee that the
following definition makes sense.

\begin{dfnn}\label{dfn star X}
For a set $X$, two sequences $(x_1, x_2, x_3, \ldots )$ and $(y_1, y_2, y_3,
\ldots )$ in $X^\integers$ are \textbf{$\mu$-equivalent} or
\textbf{equivalent} if $\mu(\{\nn: x_n = y_n\}) = 1$.
The \textbf{nonstandard
version of $X$}, denoted $\nstarback X$ and read as ``star $X$,'' is the set of
equivalence classes.
\end{dfnn}

Some examples give a sense of how much the definition contains.

\begin{enumerate}[1.]
  \item  For $X = [0, 1]$, any $r \in X$ is identified with the equivalence
class of the sequence $(r, r, r, r, \ldots)$ in $\nstarback X$.
Such points
are called \textbf{standard}, the other points in $\nstarback X$ are called
\textbf{nonstandard}.

  \item  The relation ``$<$'' is a subset of $[0, 1] \times [0, 1]$, and for
sequence $\bx$ the equivalence class of a sequence $n \mapsto x_n$ and $\by$
the equivalence class of another sequence $n \mapsto y_n$, we have $\bx
\nstarless \by$ if $\mu(\{\nn: x_n < y_n\}) = 1$.
To avoid notational
clutter, we continue to use the original symbol, ``$<$,'' for
``$\nstarless$,'' and we do the same for the other relations that we wish to
extend from $X$ to $\nstarback X$.

  \item  Using the same logic for the relations ``$\geq$,'' for any pair $\bx,
\by \in \nstar [0, 1]$, either $\bx < \by$ or $\bx \geq \by$ because for
any $A \subset \integers$, either $\mu(A) = 1$ or $\mu(A^c) = 1$.

  \item  Let $n \mapsto x_n$ be a decreasing sequence converging to $0$ and
let $dx$ denote its equivalence class.
Let $r$ denote any strictly positive
standard point in $\nstar [0, 1]$, we have $0 < dx < r$, that is, $dx$ is
\textbf{infinitesimal}, written $dx \simeq 0$.
For two numbers $\bx, \by$ in
$\nstar [0, 1]$, we write $\bx \simeq \by$ if the difference, that is, the
equivalence class of the differences, between $\bx$ and $\by$ is either a
positive or a negative infinitesimal.

  \item  Any number $\bx \in \nstar [0, 1]$ is of the form $\bx = r + dx$ for
some standard $r$ and some infinitesimal $dx$ and $r$ is called the
\textbf{standard part of $\bx$}, denoted $r = \stan(\bx)$.
A sketch of this
result is informative:  for any $m \in \integers$, partition $[0, 1]$ into,
say, finitely many half open intervals $(k/2^m, (k+1)/2^m] \cap [0, 1]$;  for
each $m$, $\bx$ is in only one of these intervals, and the limit of the Cauchy
sequence of, say, the upper endpoints is the standard number $r$;  the
equivalence class of the difference between $\bx$ and the sequence $(r, r, r,
r, \ldots)$ is necessarily infinitesimal.

  \item  Identifying a function $f:[0, 1] \rar \reals$ with its graph and
applying Definition \ref{dfn star X} to this set, for $\bx \in \nstar [0,
1]$, $\nstarback f(\bx)$ is 
the equivalence class of $(f(x_1), f(x_2), f(x_3),
\ldots )$ in $\nstar \reals$.  The function $f$ is continuous at $r \in [0,
1]$ if and only if for all infinitesimal $dx$, $f(r+dx) \simeq f(r)$, and it
has derivative $f'(r) \in \reals$ if for all non-zero infinitesimal $dx$,
$\frac{f(r+dx) - f(r)}{dx} \simeq f'(r)$.  

  \item  If $X$ is the class of finite subsets of, say, $[0, 1]$, then
$\nstarback X$ is called the class of \textbf{hyperfinite subsets of $[0,
1]$}.
Letting $H \subset \nstar [0, 1]$ be the equivalence class of the
sequence $n \mapsto F_n$ where $F_n = \{k/2^n: k = 0, 1, \ldots , 2^n\}$ gives
a hyperfinite set with $d([0, 1], H) \simeq 0$ because the equivalence class
of $n \mapsto 1/2^{n+1}$ is an infinitesimal.
\end{enumerate} 

The mapping $\bx \mapsto \sta(\bx)$ from $H$ to $[0, 1]$ is onto but one can
analyze properties of $H$ using the equivalence class of $n \mapsto F_n$.
Being a sequence of finite sets, there can be at most countably many $r \in
[0, 1]$ such that $r \in H$.
This paper makes extensive use of hyperfinite 
subsets of e.g.\ $\nstar [0, 1]$ with the exhaustiveness property that every
$r \in [0, 1]$ belongs to $H$.
For this, we need a ``larger'' version of 
$\nstarback X$.

\subsection{The Exhaustive Nets Ultrapower Construction}

Recall how we showed that exhaustive nets of finite approximations to a set
$X$ exist:  let $D = \mcalP_F(X)$ denote the class of finite subsets of a set
$X$;
for $F, F' \in D$, define $F \succsim F'$ if $F \supset F'$;
taking the
mapping from $D$ to the finite sets to be the identity mapping, we have, for
all finite $F \subset X$, there exists an $\alpha \in D$, namely $\alpha = F$,
such that for all $\beta \succsim \alpha$, $F \subset F_\beta$.
The essential
device is to find a purely finitely additive Z1, $\mu$, that puts mass $1$ on
the supersets of each and every finite subset of $X$ and then to use that
$\mu$ to define $\nstarback X$ as a set of equivalence classes.

Let $M$ be a set containing $\reals$ and the class of infinite games under
study.
Let $V_0(M) = M$, and for $n \geq 1$, let $V_n(M)$ denote the union of
$V_{n-1}(M)$ and the class of all subsets of $V_{n-1}(M)$.
Finally, let $V(M)
= \cup_{n=0}^\infty V_n(M)$.  This is called the \textbf{superstructure based
on $M$} and it contains the objects of interest.
For example, it is an
elementary exercise to show that: the equilibrium correspondence for a set of
games is an element of one of the $V_n(M)$, hence of $V(M)$;
and that for any
$Y \in V(M)$, $\mcalP_F(Y)$ also belongs to $V(M)$.

Let $X$ denote the class of all subsets of $V(M)$ and let $J$ denote the
class of all non-empty, finite subsets of $X$.
For each $a \in J$, let $J_a =
\{b \in J: a \subset b\}$, that is, $J_a$ is the collection of all finite
subsets of $X$ that contain the set $a$.
The larger is the set $a$, the
smaller is the set $J_a$.
In particular, for $a, b \in J$, $J_a \cap J_b =
J_{a \cup b}$.
We will construct $\nstarback X$ as the set of equivalence
classes in $X^J$ using a purely finitely additive Z1 much as above.

Let $\mcalF := \{A \subset J: (\exists a \in J)[J_a \subset A]\}$ denote the
collection of all subsets of $J$ containing some $J_a$ and let $\mcalJ$ denote
the class of all subsets of $J$.
We show the existence of a purely finitely
additive Z1, $\mu$, that puts mass $1$ on each element of $\mcalF$.
To this
end, let $S_{J_a}$ denote the closure of the set of point masses on supersets
of $J_a$, let $\mu$ be any element of $\bigcap_{J_a} S_{J_a}$, and define
$\nstarback X$ as the set of $\mu$-equivalence classes in $X^J$.
The
following is a restatement of \citet[Theorem 5.8, p.\
91]{hurd1985introduction}.    

\begin{thmm}\label{theorem about exhaustiveness}
If $Y$ is an element of $V(M)$ and $\mcalP_F(Y)$ is the class of finite
subsets of $Y$, then there exists an $H \in \nstarback \mcalP_F(Y)$ such that
for all $y \in Y$, $y \in H$.
\end{thmm}

\section{Proofs Omitted from the Text}

\noindent
\textbf{Proof of Lemma \ref{lemma about extensions}}.
For each finite collection $B_F \subset \mcalX^\circ$, the class of
finitely additive total probabilities, $S(B_F)$, that agrees with $q$ on $B_F$
is a non-empty, compact and convex set;
since $S(B_F) \cap S(B'_F) = S(B_F
\cup B'_F)$, the class of sets $\{S(B_F): B_F \subset \mcalX^\circ \
\mbox{finite}\,\}$ has the finite intersection property;
and the set of total
probabilities that agree with $q$ is $\bigcap S(B_F)$ where the intersection
is taken over all finite collections $B_F$.
\qed 

\noindent
\textbf{Proof of \cref{thm:equilibriaexist}}.  
Let $I_F$ and $J_F$ be finite set of agents and for $i \in I_F$, let
$B_i$ be a finite subset of $A_i$ and for $j \in J_F$, let $C_j$ be a finite
subset of $A_j$.
Define a partial order by 
\begin{equation}
(I_F, (B_i)_{i \in I_F}) \succ (J_F, (C_j)_{j \in J_F})
\end{equation}
if $J_F \subset I_F$ and for all $j \in J_F \subset I_F$, $C_j \subset B_j$.
By comprehensiveness, there exists a hyperfinite $(I_H, (H_i)_{i \in I_H})$
that is larger in the partial order than every finite $(I_F, (B_i)_{i \in
I_F})$.
Let $H = \times_{i \in I_H} H_i$.  

Pick an arbitrary $z \in \nstarback A$.
For each $b \in H = \times_{i \in
I_H} H_i$, define $v_i(b;z) = \nstarback u_i(z \backslash b)$ where for $j
\not \in I_H$, $(z \backslash b)_j = z_j$ and for $i \in I_H$, $(z
\backslash b)_i = b_i$.
By transfer of Nash's equilibrium existence theorem,
the game $\Gamma_H(z) = (H_i, v_i(\cdot;z))_{i \in I_H}$ has an equilibrium,
$(\gamma^*_i)_{i \in I_H}$.
For $j \in \nstarback I \setminus I_H$, set
$\gamma_j^*$ as point mass on $z_j$.
Let $\gamma^*$ be the product
distribution induced by $(\gamma^*_i)_{i \in \nstarback I}$ on $\nstarback A$,
define $\mu^*_i = \stan(\gamma_i)$ and $\mu^* = \stan(\gamma^*)$.
It is immediate that $\mu^*$ is a product extension of $(\mu^*_i)_{\ii}$.
And
since every $i \in I$ belong to $I_H$, and for every $i \in I$, every $b_i \in
A_i$ belongs to $H_i$, for all $i \in I$ and all $b_i \in A_i$, $u_i(\mu^*)
\geq u_i(\mu^*\backslash b_i)$.
\qed   \medskip

\noindent
\textbf{Proof of Corollary \ref{all Z1s are lof equilibria}}.
Suppose that $\mu$ is a {\rm Z1} finitely additive equilibrium for a
finite player game $\Gamma = (A_i,u_i)_{\ii}$.  

\noindent \textbf{Claim 1}.
There exists an $h^\circ \in \nstarback A$ and an
infinitesimal $\epsilon^\circ$ for which $\eta^\circ$, defined as point mass 
on $h^\circ$ is in the monad of $\mu$ and satisfies the following two
inequalities,
\begin{align}
\label{eqm utility}
  \left|
\int_A u_i(a)\,d\mu(a) 
     \;-\;  \int_{\nstarback A} \nstarback u_i(a)\,d\eta(a) \right|
& \quad<\quad \epsilon^\circ \ \hspace{10pt} \mbox{and} \\
 \mbox{for all } b_i\in A_i\quad \ 
  \left|
\int_A u_i(a \backslash b_i)\,d\mu(a)  \;- \;
   \int_{\nstarback A} 
\nstarback u_i(a \backslash b_i)\, d\eta(a) \right| & \quad<\quad \epsilon^\circ .
\label{deviations}
\end{align}
To see why, let $\mathbb{F}$ denote the class of bounded $f:A \rar \reals$.
The monad of $\mu$ is the set of all $\eta \in \nstarback \Delta(A)$ with
$|\int_A f(a)\,d\mu(a) - \int_A \nstarback f(a)\,d\eta(a)|
\simeq 0$ for all
$f \in \mathbb{F}$.  For any finite $\mathbb{F}_{Fin} \subset \mathbb{F}$ and
any $\epsilon > 0$, the set of $\eta \in \nstarback \Delta(A)$ such that for
all $f \in \mathbb{F}_{Fin}$, $|\int_A f(a)\,d\mu(a) - \int_A \nstarback
f(a)\,d\eta(a)|
< \epsilon$ is non-empty.  Since $\mu$ is a Z1, the
non-emptiness holds for the set of $\eta$ restricted to be point masses on
some $h \in \nstarback A$.
By saturation, there is an $h^\circ$ such that,
for $\eta$ being point mass on $h^\circ$, there is an exhaustive hyperfinite
$\mathbb{F}^\circ \subset \nstarback \mathbb{F}$ and an infinitesimal
$\epsilon^\circ$ such that for all $f \in \mathbb{F}^\circ \subset \nstarback
\mathbb{F}$, $|
\int_{\nstarback A} f \,d \nstarback \mu(a) - \int_{\nstarback
A} f\,d\eta(a)| < \epsilon^\circ$.
Since the exhaustive $\mathbb{F}^\circ$
contains the functions $a \mapsto u_j(a)$, $j \in I$, as well as the functions
$a \mapsto u_j(a \backslash b_i)$, we have the inequalities in equations
(\ref{eqm utility}) and (\ref{deviations}).

\noindent \textbf{Claim 2}.  There is an exhaustive hyperfinite $H = \ppi H_i$
for which point mass on $h^\circ$ is an $\epsilon^\circ$-equilibrium of the
game $(H_i, \nstarback u_i)_{\ii}$.
To see why, let $\bv = \int
u(a)\,d\mu(a)$ and for each $\ii$, let $H'_i$ be an exhaustive hyperfinite
subset of $\nstarback A_i$ that contains $h^\circ_i$.
Define $H_i$ to be
$H'_i$ with the points $h'_i$ for which $u_i(h^\circ \backslash h'_i) \geq v_i
+ \epsilon^\circ$.
This means that play if $h^\circ$ is an
$\epsilon^\circ$-equilibrium for $(H_i,u_i)_{\ii}$, and by construction, no
standard points in $H'_i$ were removed, so $H_i$ is exhaustive.
\qed  
\medskip

\noindent
\textbf{Proof of \cref{thm:closedeqmset}}.  
Non-emptiness is given by \cref{thm:equilibriaexist}.
Given
the compactness of the range space for the correspondence, $\Delta$, it is
sufficient to show that the graph of $u \mapsto Eq(u)$ is closed.
Let
$p^\alpha \in Eq(u^\alpha)$ with $p^\alpha \rar p$ and $u^\alpha \rar u$.  
We must show that $p \in Eq(u)$.

Suppose, for the purpose of establishing a contradiction, that $p$ is not an
equilibrium of $\Gamma(u)$.
This requires that for some $\ii$ and some $b_i
\in A_i$, there exists a strictly positive $r$ such that 
\[  (\ddagger)\ \  \int u_i(a)\,dp 
                 =  \int u_i(a \backslash b_i) - r.\]  
From the triangle inequality,
\begin{equation}
\left|
\int u_i^\alpha \,dp^\alpha  - 
     \int u_i\,dp \right| \leq 
\left|
\int u_i^\alpha \,dp^\alpha - 
    \int u_i\,dp^\alpha \right| + 
  \left|
\int u_i \,dp^\alpha - 
        \int u_i\,dp \right| .
\end{equation}
The first term on the right goes to $0$ because $\|u_i^\alpha - u_i\|$ goes to
$0$ and the second term goes to $0$ because  $p^\alpha \rar p$.
By the same
argument,
\begin{equation}
\left|  \int u_i^\alpha(a \backslash b_i) \,dp^\alpha(a) - 
         \int u_i(a \backslash b_i)\,dp(a) \right|
\rar 0  .
\end{equation}
Thus, there exists an $\alpha$ such that for all $\beta \succsim \alpha$, both
differences are smaller than $r/2$, a contradiction to $(\ddagger)$.
\qed
\medskip

\noindent \textbf{Proof of \cref{cor:caepsilonequilibria}}. For any net $\epsilon^\alpha \rar 0$, let $p^\alpha$ be a finitely or
countably additive $\epsilon^\alpha$-equilibrium, let $p$ be an accumulation
point, and let $\beta \mapsto (p^\beta, \epsilon^\beta)$ be a subnet along
which we have convergence to $(p,0)$.
By definition of convergence of
finitely additive probabilities, the utilities of the $\epsilon^\beta$
equilibria converge to the utilities associated with the strategy $p$, and the
same is true for any pure strategy deviation.
\qed  
\medskip

\noindent \textbf{Proof of \cref{thm:iterundtd}}.  Let $F = \ppi F_i$ be a
product of non-empty finite subsets of $\ppi A_i$.
The class of hyperfinite
products $H = \ppi H_i$ containing $F$ is internal.
For each such $H$, the
set of equilibria in iteratively undominated strategies is internal.
The
union of these internal sets is itself an internal subset of $\nstarback
\Delta$, and the standard part of any internal set is closed in the weak$^*$
topology on finitely additive probabilities.
Let $Un(F)$ denote that closed
set in $\Delta$.  The class $\{Un(F): F = \ppi F_i\}$ has the finite
intersection property, hence has non-empty, closed intersection, $Un$.
By
construction, any element of $Un$ puts mass $0$ on the set of weakly
undominated strategies, and it is the standard part of an equilibrium for some
hyperfinite version of the game.
\qed  \medskip

\noindent
\textbf{Proof of Lemma \ref{Simon-Zame lemma}}.  
We start with the last assertion.
In any finite approximation
$(F_{1,\alpha},F_{2,\alpha})$ containing $(0.8, 0.8)$, let $h_1^\alpha \in
F_1^\alpha$ be the largest element smaller than $0.8$.
Any $h'_1 <
h_1^\alpha$ in $F_1^\alpha$ is strictly dominated for $1$ and any $h_2 > 0.8$
in $F_2^\alpha$ is strictly dominated by $0.8$ for player $2$.
After
eliminating these strategies in the game played on $F_1^\alpha \times
F_2^\alpha$, the unique equilibrium is point mass on the pair $(h_1^\alpha,
0.8)$.
Taking limits as $F_1^\alpha$ becomes exhaustive guarantees that for
all $\epsilon > 0$, there exists an $\alpha$ such that for all $\beta \succsim
\alpha$, $h_1^\beta$ is in the interval $(0.8-\epsilon, 0.8)$.
The limit Z1's
must have $\mu_1^*((0.8-\epsilon,0.8)) \equiv 1$ while $\mu_2^*$ is the
countably additive point mass on $0.8$.

For the initial claims, verifying that any product extension of
$(\mu_1^*,\mu_2^*)$ is an equilibrium if and only if $(\dagger)$ and
$(\ddagger)$ hold is immediate.
Fix any equilibrium $\mu^*$ and any pair of
exhaustive nets $\alpha \mapsto (F_1^\alpha, F_2^\alpha)$.
By Lemma
\ref{lemma characterizing exhaustive nets}, there are distributions
$\eta^\alpha$ on the exhaustive net $\alpha \mapsto F_1^\alpha \times
F_2^\alpha$ converging to $\mu^*$.
Since $\mu^*$ is a product extension, the
distribution $\eta^\alpha$ can be taken to be a product measure $\eta_1^\alpha
\times \eta_2^\alpha$.
The pair is an $\epsilon$-equilibrium if and only if
$\eta_1^\alpha$ puts nearly no mass on the sets $[0, 0.8-\epsilon]$ or
$\{0.8\}$ while $\eta_2^\alpha$ puts nearly no mass on the sets
$[0.8+\epsilon, 1]$ and this happens if and only if $(\dagger)$ and
$(\ddagger)$ hold.
\qed  \medskip

We now turn to the analysis of the symmetric hyperfinite version of the
Sion-Wolfe game.
\medskip

\noindent \textbf{Proof of Lemma \ref{Sion-Wolfe lemma equal Hs}}.  Every $h_1
\in H_1 \cap (\half, 1)$ is weakly dominated by $a_1 = 1$, and every $h_2 \in
H_2 \cap [0, h')$ is weakly dominated by $h'$.
After these strategies are
eliminated from the game, every $h_1 \in (0, h') \cap H_1$ is weakly dominated
by $a_1 = 0$, and $0.5$ is weakly dominated by $a_1 = 1$.
When $1$ is only
using the strategies $0, h'$, and $1$, the only weakly undominated strategies
for player $2$ are the strategies $h', 0.5$, and $1$.
The payoffs in the
resultant $3 \times 3$ game are given by

\begin{center}
\begin{tabular}{lc|c|c|c|}  
            & \multicolumn{4}{c}{Player $2$}                        \\
            &        & $h'$      &  $0.5$ & $1$  \\      \cline{2-5}       
            &  $0$   & $(-1,+1)$ &  $(0,0)$  & $(+1,-1)$ \\   \cline{2-5}       
Player $1$  &  $h'$  & $(0,0)$   & $(-1,+1)$ & $(+1,-1)$ \\      \cline{2-5}  
            &  $1$   & $(+1,-1)$ & $(+1,-1)$ & $(0,0)$   \\      \cline{2-5}  
\end{tabular}  
\end{center}
\medskip
Direct examination shows that there is no pure strategy equilibrium, and 
that
the unique distribution for player $2$'s actions that makes $1$ indifferent is
$(\frac{1}{5}, \frac{1}{5}, \frac{3}{5})$ on $(h',0.5,1)$, and that the unique
distribution for player $1$'s actions that makes $2$ indifferent is
$(\frac{1}{5}, \frac{1}{5}, \frac{3}{5})$ on $(0, h', 1)$.
\qed  \medskip

The following is the analysis of the asymmetric version of the Sion-Wolfe game
described in the text.
\medskip

\noindent
\textbf{Proof of Lemma \ref{Sion-Wolfe lemma}}.  
Verifying that $\mu^*$ is an equilibrium was done in the text.
For the rest,
let $H_1$ and $H_2$ be exhaustive hyperfinite subsets of $A_1$ and $A_2$
respectively with the property that $h_1' < h_2' < 0.5$.
\smallskip \begin{center}First round of deletion of weakly dominated
strategies \end{center} \smallskip

First, observe that any $h_1 \in H_1 \cap (0.5 ,1)$ is weakly dominated by $1$
for player $1$ because moving to a higher strategy in $(0.5, 1]$ wins against
every $h_2$ that a lower strategy beats, and either wins or ties against every
$h_2$ that a lower strategy loses to.

Second, observe that any $h_2 \in [0, h_2') \cap H_2$ is weakly dominated by
$h_2'$ for player $2$ because moving to a higher strategy in $[0, h_2')$ wins
against every $h_1$ that a lower strategy beats, and either wins or ties
against every $h_1$ that a lower strategy 
loses to.

After deleting the weakly dominated strategies, the action sets for the two
players are $(H_1 \cap [0, 0.5]) \cup \{1\}$ for player $1$ and 
$H_2 \cap [h_2', 1]$ for player $2$.
\smallskip \begin{center}Second round of deletion of weakly dominated
strategies \end{center} \smallskip

Now consider the game with the weakly dominated strategies deleted.
For
player $1$, playing $x = 0$ weakly dominates $0 < h_1 < \half$ and playing $x
= 1$ weakly dominates $\half$.
For player $2$, the only weakly undominated
strategies are $y = h_2'$ and $y = 1$.
With the weakly dominated strategies deleted, we have the $2 \times 2$ game
given by \medskip

\begin{center}
\begin{tabular}{lc|c|c|}  
           & \multicolumn{3}{c}{Player $2$}                        \\
           &              & $y = h_2'$ &  $y = 1$ \\      \cline{2-4}       
Player $1$ &    $x = 0$   & (-1,+1)    &  (+1,-1)\\      \cline{2-4}       
           &    $x = 1$   & (+1,-1)    &  (0,0)\\      \cline{2-4}  
\end{tabular} \end{center} \smallskip
Direct verification shows that $1$ playing $(\frac{2}{3}, \frac{1}{3})$ on
$x = 0$ and $x = 1$ and $2$ playing $(\frac{1}{3}, \frac{2}{3})$ on $y = h'_2$
and $y = 1$ is the unique equilibrium.
\qed  \medskip

\noindent \textbf{Proof of Lemma \ref{lemma about the Wald game}}.  Each
$\mu_i$ is a convex combination of Z1's.
From \citet[Corollary
5.2]{stinchcombe2023direct}, for any exhaustive hyperfinite $H_i$, there are
distributions $\eta_i$ on $H_i$ agreeing with $\mu_i$ for all measurable
sets that put arbitrary weights on pairs $(h_1,h_2)$ with $h_1 > h_2$ and
the reverse.
These give equilibria with the given range of payoffs, and no 
devation $b_i \in A_i$ is profitable.

For any exhaustive hyperfinite $H_1$ and $H_2$, there are two cases to
consider:  one of the players has a larger element in $H_i$;
or the largest
elements are equal.  In the first case, the approximate equilibria necessarily
involve the player with the larger element(s) putting mass infinitesimally
close to $1$ on the larger element(s), and the associated payoffs are either
$(+1,-1)$ or $(-1, +1)$.
In the second case, in any equilibrium, both players
necessarily put mass infinitesimally close to $1$ on their largest element,
and the associated payoffs or $(0, 0)$.
\qed  \medskip

\noindent \textbf{Proof of Lemma \ref{lemma on Karlin games}}.
It can be
checked that any action in $[0, \infty)$ strictly dominates $\infty$ for both
players.
For the game played on any exhaustive hyperfinite pair of actions
sets $H_1$ and $H_2$, it is an equilibrium for both to player their
(necessarily infinite) largest actions in $\nstar [0, \infty)$.
Depending on
the difference between these two largest elements, the standard part of the
payoffs is any pair $(r,-r)$ for $r \in [-1, +1]$.

For any $p \in \Delta(X)$, the function $y \mapsto \int_X u_2(x,y)\,dp(x)$ is
strictly increasing on $[0, \infty)$ so that any $q^*$ must satisfy
$q^*((n,\infty)) \equiv 1$.
If $q^*((n,\infty)) \equiv 1$, then the function
$x \mapsto \int_Y u_1(x,y)\,dq^*(y)$ has supremum $-1$, and any $p^*$
satisfying $p^((n,\infty)) \equiv 1$ delivers this payoff.
\qed  

The following concerns the integrals of semi-continuous functions against
continuously equivalent probabilities.  \medskip

\noindent \textbf{Proof of Lemma \ref{lemma for deviations}}.
It is
sufficient to show that the inequality holds for $f(x) = 1_F(x)$, $F$ a closed
subset of $X$ (because every non-negative upper semi-continuous function is a
uniform limit of positive linear combinations of indicators of closed sets).
With $F^{1/n}$ denoting the open set of points $y$ with $d(y,F) < 1/n$, we
have $A_n := (F^{1/n} \setminus F) \downarrow \emptyset$ so $q(A_n) \downarrow
0$, but since $p$ is finitely additive, $\inf_n p(A_n) > 0$ is
possible.
If this happens, $\int 1_F\,dq > \int 1_F\,dp$.  And since $f$ is
upper semi-continuous iff $-f$ is lower semi-continuous, the inequality
reverses for lower semi-continuous functions.
\qed  \medskip

The following concerns the integrals finitely additive continuous equivalents
of a countably additive probability that avoids discontinuities.
\medskip

\noindent \textbf{Proof of Lemma \ref{lemma for the same utility}}.  
Rescaling if necessary, we can, without loss, assume that $f(x) \in [-1,
+1]$.
Let $F$ denote the closure of the set of discontinuities of $f$,
suppose that $q = ca(p)$ is the countably additive version of a finitely
additive total probability $p$, and that $q(F) = 0$.
Let $G$ denote the open
complement of $F$, and pick arbitrary $\epsilon > 0$.
We will show that
$|\int f\,dp - \int f\,dq| < \epsilon$.

Pick $K \subset G$ such that $q(K) > 1 - \epsilon/4$.
By the continuous
equivalence of $p$ and $q$, for any $\delta > 0$, $p(K^\delta) > 1 -
\epsilon/4$.
Pick $\delta > 0$ such that $K^{2\delta} \subset G$.  Let $g$
denote the restriction of $f$ to the closure of $K^\delta$.
By the usual
results on the extension of continuous functions defined on closed sets, $g$
has a continuous extension, $h$, to all of $X$ with $\|h\|
\leq \|g\|$.  Since
$p$ and $q$ are continuously equivalent, $\int h\,dp = \int h\,dq$.  We have
\[
\Biggl|
\int f\,dp\,  -  
            \int f\,dq \,\Biggr|
\leq 
  \Biggl|  \int f\,dp\, -  \int h\,dp\, \Biggr|  + 
  \Biggl|
\int h\,dp\, -  \int h\,dq\, \Biggr| + 
   \Biggl|  \int h\,dq\, -\int f\,dq\, \Biggr|.
\nonumber
\]
The middle term is equal to $0$ by continuous equivalence.
The first and the
third terms are less than $\epsilon/2$ because the functions $f$ and $h$ agree
on $K^\delta$, a set that both probabilities assign at least mass $1 -
\epsilon/4$, and the absolute value of the difference between $f$ and $h$ is
bounded by $2$ because both take values only in $[-1, +1]$.
\qed  \medskip

\noindent \textbf{Proof of \cref{thm:ctsequivofacaeqm}}.  By Lemma
\ref{lemma for the same utility}, any continuously equivalent $p$ satisfies
$\int u\,dq^* = \int u\,dp$.
Since we also know that $q^*$ is an equilibrium,
for all $\ii$ and all $b_i \in A_i$, we have \begin{equation}  \int
u_i(a)\,dp(a) = \int u_i(a)\,dq^*(a) \geq  \int u_i(a \backslash
b_i)\,dq^*(a).
\end{equation} Pick arbitrary $\ii$ and $b_i \in A_i$.  By
assumption, $a \mapsto u_i(a \backslash b_i)$ is upper semi-continuous.
By
Lemma \ref{lemma for deviations}, $\int u_i(a \backslash b_i)\,dq^*(a) \geq
\int u_i(a \backslash b_i)\,dp(a)$.
Combining, for all $\ii$ and all $b_i \in
A_i$, $\int u_i(a)\,dp \geq \int u_i(a \backslash b_i)\,dp$ so that $p$ is an
equilibrium.
\qed  \medskip

The omitted proof of \cref{thm:ctsequivofafaeqm} is a mirror image
of the proof for \cref{thm:ctsequivofacaeqm} using the lower
semi-continuous part of Lemma \ref{lemma for the same utility}.

\noindent
\textbf{Proof of Lemma \ref{lemma product sigma fields don't cut it}}.
There is no loss in assuming that $\mcalX$ is the class of all subsets
of $X$ since this will have the largest product $\sigma$-field.
For
\emph{reductio ad absurdum}, assume that the diagonal belongs to $\mcalX
\otimes \mcalX$.
Then there is a countable family $F := \bigl(A_{m}\times
B_{m}:m\in\mathbb{N}\bigr)$ of subsets of $X\times X$ such that the diagonal
belongs to the $\sigma$-field generated by $F$.
Now for each
$s\in\{0,1\}^{\mathbb N}$ and $m\in\mathbb{N}$, define: 
\[  D_m(s):=
\begin{cases} A_m & \text{if } s_m=1,\\ 
A_m^{c} & \text{if } s_m=0,
\end{cases}\quad\qquad 
 C_s \;:=\;
\bigcap_{n\in\mathbb N} D_n(s). \]  
Likewise define each $E_t$ from the $B_m$ for $m\in \mathbb{N}$.
Then the
rectangles $C_s\times E_t$ for  $(s,t)\in\{0,1\}^{\mathbb
N}\times\{0,1\}^{\mathbb N}$ form a partition of $X\times X$ into at most
$|\{0,1\}^{\mathbb N}\times\{0,1\}^{\mathbb N}|=2^{\aleph_0}$ cells, so since
$\sigma(F)$ consists exactly of unions of subcollections of these cells, the
diagonal is a union of a family $\mathcal{F}$ of at most $2^{\aleph_0}$
rectangles.
But since the cardinality of the diagonal is equal to the
cardinality of $X$, denote $|X|$, and
$|X|>|\{0,1\}^{\mathbb{N}}|=2^{\aleph_0}$, it follows that there are points
$x,y\in X$ with $x\neq y$ and rectangle $A\times B\in\mathcal{F}$ such that
$x,y\in A$ and $x,y\in B$, whence  $(x,y)\in A\times B\subseteq \Delta$, which
is impossible.
\qed

\end{document}